\documentclass[11pt]{amsart}

\usepackage[reset,a4paper,vmargin=3truecm,hmargin=2.5truecm]{geometry}
\usepackage{amsmath,amsthm,amssymb,amscd,amsthm,,epic}
\usepackage{bbm,graphicx,bm}
\usepackage{mathtools}
\usepackage{leftidx}

\usepackage{subfig}
\usepackage{hyperref}

\usepackage{tikz}
\usetikzlibrary{shapes.geometric}
\usetikzlibrary{cd}
\usetikzlibrary{positioning}
\tikzcdset{every label/.append style = {font = \scriptsize}}

\usepackage{ctable,tabularx}

\usepackage{color}

\theoremstyle{plain}
\newtheorem{thm}{Theorem}[section]

\newtheorem{prop}[thm]{Proposition}
\newtheorem{cor}[thm]{Corollary}

\theoremstyle{definition}
\newtheorem{dfn}{Definition}[section]
\newtheorem{ex}[thm]{Example}

\theoremstyle{remark}
\newtheorem{rem}{Remark}[section]
\DeclareMathOperator{\Spec}{Spec}

\makeatletter

\newcommand{\HRabi}{H_{\text{\upshape Rabi}}}
\newcommand{\AHRabi}[1]{H_{\text{\upshape Rabi}}^{#1}} 
\newcommand{\cp}[2]{P^{(#1)}_{#2}} 

\newcommand{\Z}{\mathbb{Z}} 
\newcommand{\R}{\mathbb{R}} 
\newcommand{\C}{\mathbb{C}} 

\newcommand{\bA}{\mathbf{A}}
\newcommand{\bJ}{\mathbf{J}}
\newcommand{\bI}{\mathbf{I}}
\newcommand{\bK}{\mathbf{K}}
\newcommand{\bC}{\mathbf{C}}
\newcommand{\bM}{\mathbf{M}}

\newcommand{\sh}{{\rm sh}}
\newcommand{\ch}{{\rm ch}}
\newcommand{\tah}{{\rm th}}
\newcommand{\cth}{{\rm cth}}
\newcommand{\sch}{{\rm sech}}
\newcommand{\e}{\varepsilon}

\DeclareMathOperator{\tridiag}{tridiag}
\newcommand{\mat}[1]{\begin{bmatrix}#1\end{bmatrix}}
\newcommand{\Tridiag}[4]{\tridiag\mat{#1 & #2 \\ #3}_{#4}}

\newcommand{\sNCHO}[1]{Q^{(#1)}}

\title[Coverings of AQRM: $\eta$-NCHO]
{Covering families 
of the asymmetric quantum Rabi model: $\eta$-shifted non-commutative harmonic oscillators}

\author{Cid Reyes-Bustos and Masato Wakayama}

\begin{document}

\begin{abstract}
  The non-commutative harmonic oscillator (NCHO) is a matrix valued differential operator originally introduced
  as a generalization of the quantum harmonic oscillator having a weaker $\mathfrak{sl}_2(\R)$-symmetry. 
  The spectrum of the NCHO has remarkable properties, including the presence of number theoretical structures
  such as modular forms, elliptic curves and Eichler cohomology observed in the special values of the associated
  spectral zeta function. 
  In addition, the Heun ODE picture of the eigenvalue problem of the NCHO reveals a connection with the quantum
  Rabi model (QRM), a fundamental interaction model from quantum optics. In this paper we introduce an $\eta$-shifted
  NCHO ($\eta$-NCHO) that has an analogous relation with the asymmetric quantum Rabi model (AQRM)
  which breaks an apparent symmetry of the QRM,  and describe its basic properties. 
  Even though the shift factor does not break the parity symmetry of the NCHO, a
  certain type of degeneracies appears for $\eta \in \frac12 \Z$, as if mirroring the situation of the AQRM.
  We give furthermore a detailed description of the confluence process, that we call {\it iso-parallel} confluence
  process due to the fact that it requires a  parallel transformation of two parameters describing the spectrum
  of  $\eta$-NCHO and representations of $\mathfrak{sl}_2(\R)$. We relate the eigenvalues of the two
  models under the iso-parallel confluence process, including how the quasi-exact eigenfunctions of the $\eta$-NCHO correspond
  to Juddian solutions of the AQRM. From the point of view of this confluence process, a family of $\eta$-NCHO corresponds to a single AQRM, thus we may regard the $\eta$-NCHO as a covering of the AQRM. We expect the study of the $\eta$-NCHO and the AQRM from this point of view to be helpful for the clarification of several questions on the AQRM, including the hidden symmetry and the number of Juddian solutions.
  
 \,
 

 \smallskip


\end{abstract}

\date{\today}

\maketitle

\tableofcontents
  
\section{Introduction}
\label{sec:introduction-1}

The mathematical model known as the non-commutative harmonic oscillator (NCHO) is defined as a matrix-valued ordinary differential operator
or  Hamiltonian, acting on $L^2(\R)\otimes \C^2$,  given by
\begin{equation}
  \label{eq:Qdef}
  Q :=  \bA \left( -\frac12 \frac{d^2}{d x^2}   + \frac12 x^2 \right)  +  \bJ \left( x \frac{d}{d x}  + \frac12 \right) ,
\end{equation}
for two parameter $\alpha,\beta \in \R$ such that $\alpha\beta>1$, where 
\[
  \bA =
  \begin{bmatrix}
    \alpha & 0 \\
    0 & \beta
  \end{bmatrix}, \qquad
    \bJ =
  \begin{bmatrix}
    0 & -1 \\
    1 & 0
  \end{bmatrix}.
\]
The NCHO was introduced by Alberto Parmeggiani and the second author as a non-commutative generalization of the usual quantum harmonic oscillator in \cite{PW2001}. Actually, in addition to the canonical commutation relation $[\frac{d}{d x}, x]=1$ (CCR),  there is another non-commutativity involving the matrices appearing in the Hamiltonian. If $\alpha=\beta$, the NCHO is unitarily equivalent to a couple of the harmonic oscillator \cite{PW2002}. The NCHO has been studied from diverse points of view, including the description of the spectrum using semi-classical analysis \cite{P2008}, applications to PDE via generalizations of inequalities such as the Fefferman-Phong inequality (see \cite{P2014Milan} and the references therein), and the arithmetics in the spectrum and its associated spectral zeta function \cite{KW2020}. In particular, the study of the special values of the spectral zeta function of the NCHO has revealed a rich arithmetic theory including elliptic curves, modular forms and natural extensions of Eichler forms with associated cohomology group.
It is actually an extended study of the special values of Riemann zeta function $\zeta(s)$, particularly of the irrationality proof of $\zeta(3)$ initiated by Roger Ap\'ery and the deeper study from the point of view of algebraic geometry and modular forms by F. Beukers (See \cite{KW2020} and the references therein).  

The quantum Rabi model (QRM) describes the fundamental interaction between a photon and a two-level atom. The Hamiltonian $\HRabi$ is given by 
\begin{equation*}
  \HRabi = \omega a^\dag a+\Delta \sigma_z +g \sigma_x(a^\dag+a) 
\end{equation*}
where $a^\dag$ and $a$ are the creation and annihilation operators of the bosonic mode, i.e. $[a,\,a^\dag]=1$ (CCR) and
\[
\sigma_x = \begin{bmatrix}
 0 & 1  \\
 1 & 0
\end{bmatrix},
\quad
\sigma_y= \begin{bmatrix}
 0 & -i  \\
 i & 0
\end{bmatrix}  
\quad
 \sigma_z= \begin{bmatrix}
 1 & 0  \\
 0 & -1
\end{bmatrix}  
\]
are the Pauli matrices, $2\Delta$ is the energy difference between the two levels, $g$ denotes the coupling strength between the two-level system and the bosonic mode with frequency $\omega$ (subsequently, we set $\omega=1$ without loss of generality). The QRM is the fully quantized version \cite{JC1963} of the original semi-classical Rabi model \cite{R1936} and it is widely recognized as the most fundamental model for the study of quantum interaction \cite{bcbs2016}. Because of this, the QRM has applications for quantum optics, solid state physics and potential applications to quantum information theory beyond theoretical studies (see e.g., \cite{bcbs2016,HR2008, Y2017, YS2018}). In addition to experimental and theoretical physics, the QRM and its generalizations have been the subject of considerable mathematical research (see e.g. \cite{BdMZ2019, HH2012, KRW2017, RW2019, RBW2022,Sugi2016}).

In general, the QRM has spectral crossings (spectral degeneracies) and a natural $\Z_2$-symmetry (parity) that allows the labeling of eigenvalues in such a way that crossings occur between eigenvalues of opposite parities (this has been related to the notion of integrability in quantum systems, see e.g. the discussion in \cite{B2011PRL}). We note that in \cite{HH2012} it was proved that the ground state of the QRM is always non-degenerate by showing the ground state consists of even functions \cite{W2013}.

In \cite{W2015IMRN} it was shown that the Heun ODE picture of the NCHO, with the 4 regular singular points $\omega=0,1,\alpha\beta$ and $\infty$, and the confluent Heun ODE picture of the QRM are related by a confluence process making the regular singular points $\alpha\beta(=\det{ \bA})$ and $\infty$ coalesce. The confluence process used for the NCHO involves a particular meaning for the coalescence of the singular points.
In fact, it requires certain changes of variables in the system and representation of $\mathfrak{sl}_2(\R)$ (via a parameter that defines the representation, see  Section \ref{sec:cprod}). The general picture of the process is shown in the diagram in Figure \ref{fig:conf}, however, we remark that at that point the properties and the nature of the confluence process remained mysterious.

\begin{figure}[h]
  \centering
  \begin{tikzcd}
    \text{NCHO} &[1.5em] \arrow[l,"{\substack{\pi' \\ (\cong \pi'_a (a=1,2))}}","{\substack{\text{Oscillator}\\ \text{representation} \\ \text{}} }"']  \underset{\text{(degree 2)}}{\mathcal{R} \in \mathcal{U}(\mathfrak{sl}_2)} \arrow[r,"{\substack{ \text{Generalized Laplace} \\ \text{transform } {\mathcal L}_a\\ \text{(Intertwiner)} \\ \text{} }}", "{ \substack{\pi_a' \cong \varpi_a \\ \text{(non-unitary} \\ \text{principal series)} }}"']  & \text{\;Heun ODE\;} \arrow[d,Rightarrow,"{\substack{ \text{confluence}\\ \text{process}}}"] \\
    &  \underset{\text{(degree 2)}}{\mathcal{K} \in \mathcal{U}(\mathfrak{sl}_2)} \arrow[r, shorten >= 2mm, controls={+(3.5,-3) and +(-0.05,-0.3)},"\varpi_{\tilde{a}}" near start] & \text{Confluent Heun ODE} \\
    & & \text{QRM} \arrow[u,"{\substack{\text{Segal-Bargmann space}\\  \text{representation} \\ \text{} }}"']
  \end{tikzcd}
  \caption{Relation between the QRM and the NCHO via confluence process}
  \label{fig:conf}
\end{figure}
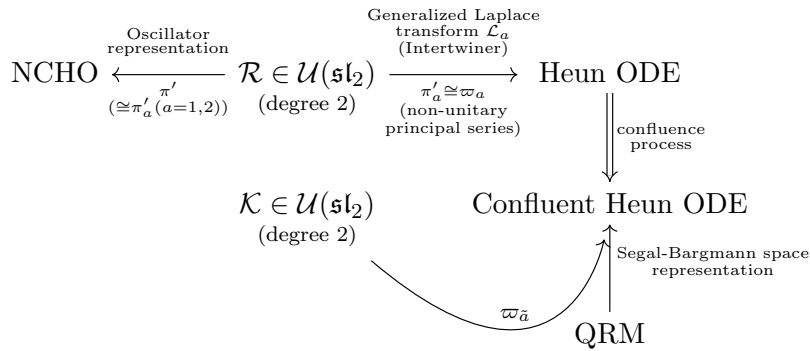

The introduction of a biased term $ \eta \sigma_x , \, (\eta \in \R)$ to the QRM Hamiltonian results on the model known as the asymmetric quantum Rabi model (AQRM) with Hamiltonian
\[
  \AHRabi{\eta} = \omega a^\dag a+\Delta \sigma_z +g\sigma_x(a^\dag+a) + \eta \sigma_x,
\]
The additional term breaks the symmetry of the QRM Hamiltonian and, in general, makes the spectrum multiplicity free. For $\eta \in \frac12 \Z $ the degeneracies appear again despite the lack of an ``evident'' symmetry operator like in the case of the usual symmetry QRM, i.e. $\eta=0$. We refer the reader to \cite{LB2015JPA,LB2016JPA,W2016JPA} and in particular to \cite{KRW2017} for a complete picture of the degenerations in the spectrum of the AQRM and to \cite{BLZ2015} for an alternate point of view of the degeneracies as conical intersections in the energy landscape described by the AQRM. 

The presence of spectral degeneracies in the spectrum of the AQRM for the case $\eta \in \frac12 \Z$ was widely believed to be caused by a {\it hidden symmetry}, that is, by an operator commuting with the Hamiltonian. In \cite{A2020} based on numerical computations it was argued that if such a symmetry operator exists then it must depend non-trivially on the system parameters. For $\eta \in \frac12 \Z$, the existence of the hidden symmetry operator $J_{2 \eta}$ satisfying
\[
  [\AHRabi{\eta},J_{2\eta}] = 0
\]
was first heuristically demonstrated in the interesting trial paper\cite{MBB2020} based on an idea from classical mechanics integrability, that is, by assuming the existence of operators that commute with the Hamiltonian $ \AHRabi{\eta} $ (see cf. \cite{CM2011}). The existence of the hidden symmetry operator $J_{2 \eta}$ was later proved in  \cite{RBW2021} in full generality. The properties of the operator $J_{2 \eta}$  reveal a rich geometric theory behind the relation between the hidden symmetry of the AQRM and the constraint polynomials appearing by the study  of the degeneracy of the eigenvalues \cite{RBW2022}. 

Therefore, it is natural to consider a generalization of the NCHO such that the confluence process for two regular singular points $\alpha\beta$ and $\infty$ of its Heun ODE picture gives the confluent Heun ODE picture of the AQRM. We call the resulting model the $\eta$-shifted NCHO ($\eta$-NCHO)
with Hamiltonian given by
\[
  Q^{(\eta)} = \bA \left( -\frac12 \frac{d^2}{d x^2}  + \frac12 x^2 \right) +  \bJ \left( x \frac{d}{d x} + \frac12 \right) + 2 \eta i \det(\bA \pm i\bJ)^{\frac12} \bJ,
\]
for $\eta \in \R$. It is important to note that while in the $\eta$-NCHO the symmetry is not broken (i.e. $Q^{(\eta)}$ is of parity preserving for any $\eta \in \R$), certain degeneracies in the spectrum disappear except in the case of half-integer parameter $\eta$ (i.e. $\eta \in \frac12 \Z$), resembling the situation of the AQRM. In the $\eta$-NCHO case the spectral structure is revealed in the study of the monodromy representations of the singularities of the corresponding Heun ODE.

In this paper, we name the confluence process that relates the $\eta$-NCHO the {\em the iso-parallel confluence process} because of the fact that it requires a simultaneous parallel displacement of the parameters of the Heun ODE pictures and the associated representation of $\mathfrak{sl}_2$ (see Definition \ref{dfn:ipc}). An additional feature of the iso-parallel confluence process is that it contains also the asymptotic information of the growth of the parameters $\alpha,\beta$ of the $\eta$-NCHO as the regular singularity at $\omega= \alpha \beta$ coalesces into the the singularity $\omega = \infty$. The asymptotic data given by two auxiliary parameters $k,r$ turns out to be crucial as it allows to identify the two models in the confluence picture. A remarkable result  obtained from the iso-parallel confluent process is a very precise relation between the eigenvalues of finite type of the $\eta$-NCHO and the Juddian eigenvalues of the AQRM (see Section \ref{sec:cpoly}) and the parameters $k,r$ are related to the system parameters $g,\Delta$ of the AQRM via certain compatibility/constraint equations given by the confluence process.

The fact that the asymptotic parameters $r, k$ of the iso-parallel confluence process determine a specific AQRM model given by the system parameters $g, \Delta$ suggests that in fact we may regard the $\eta$-NCHO as a {\it covering} of the AQRM. More precisely, a specific AQRM is related to a family of $\eta$-NCHO via the iso-parallel confluence process and we say that the $\eta$-NCHO is a covering model of the AQRM.  The iso-parallel confluence remains a mysterious operation, and understanding its mathematical and physical implications may be a significant problem to be solved in the future, e.g. obtaining a deeper understanding of monodromy problems for the (confluent) Heun ODE. In particular, it is fundamental to investigate potential covering relations between generalizations of the QRM like the Dicke model, (finite-type) Spin-Boson model or the quantum Rabi-Stark model, and corresponding variations of the NCHO. Another interesting problem is finding explicit relations between the partition functions and/or the spectral zeta functions (Dirichlet series of eigenvalues) of the ($\eta$-)NCHO and the (A)QRM (correspondingly a model and its covering model) or its special values.  We refer the reader to \cite{IW2005a} for the analytic continuation for the spectral zeta function of the NCHO and \cite{Sugi2016} for the QRM using the same method. A different method using a contour integral expression, like the one for the Riemann zeta function $\zeta(s)$, using the partition function derived from the explicit analytical formula for the heat kernel of the QRM \cite{RW2019} is given in \cite{RW2020z} (see \cite{R2020} for the case of the  AQRM).
Moreover, the conical intersection \cite{BLZ2015} may exhibit the geometric phase (Berry phase) \cite{B1984} which appears by considering the trajectories encircling the intersection points. Clarifying the correspondence between the conical interactions between the AQRM and $\eta$-NCHO may be important for the understanding the covering of the model as a covering map of Riemann surfaces.

Let us describe the general organization of the paper. The paper is divided in three parts.
In Part \ref{sec:partsNCHO} we introduce the $\eta$-NCHO and describe its basic properties. Concretely, in Section \ref{sec:biased-ncho-} we define the $\eta$-shifted NCHO by its Hamiltonian and describe its basic properties. In order to study the spectral structure and the confluence picture, we develop the Heun picture of the $\eta$-NCHO in Section \ref{sec:HeunPicture}. The basic idea describing the eigenvalue problem of the NCHO by Fuchsian differential equations and the connection problem for the corresponding  Heun ODE can be found in \cite{O2001, O2004} for odd eigenfunctions and  \cite{W2015IMRN} for even eigenfunction. The spectral degeneracy structure of the $\eta$-NCHO is described in Section \ref{sec:multiplicity}, where we show the presence or absence of certain degeneracies depending on the parameter $\eta$ using the monodromy representation.

In Part \ref{sec:fromsNCHOtoAQRM} we describe the confluence process that relates the $\eta$-NCHO with the AQRM. In Section \ref{sec:preliminaries} we introduce the AQRM and its confluent Heun picture. In Section \ref{sec:cprod} we define the iso-parallel confluence process
and give the main result of this paper, that the iso-parallel process connects the ODE pictures of the two models. In addition, we explicitly show the relations between the parameters of the two models given by certain additional constraint equations.

In Part \ref{sec:extendeddiscss} we describe some of the consequences of the iso-parallel confluence process focused on the quasi-exact solutions of both models. In Sections \ref{sec:cpoly} and \ref{sec:cpoly2} we explicitly describe how the quasi-exact solutions of the $\eta$-NCHO descend into Juddian eigenvalues of the AQRM, both in the general case and the case $\eta \in \frac12 \Z$ and also the converse situation, that is, the case where a Juddian solution descends from a quasi-exact solution of the $\eta$-NCHO. We recall the significance of the Juddian solutions in the AQRM in Section \ref{sec:cpolysig}, recalling how the Juddian solutions in the AQRM actually determine the general features of the spectral curves. Finally, in Section \ref{sec:comparison} we give a summary of the spectral structures $\eta$-NCHO and the AQRM, including their degeneracy picture and associated symmetries. From the point of view of the confluence picture, the hidden symmetry of the AQRM appears as a reflection of the symmetry of $\eta$-NCHO.

\part{The $\eta$-shifted NCHO}
\label{sec:partsNCHO}

In this part we introduce the $\eta$-shifted non-commutative harmonic oscillator ($\eta$-NCHO) and describe the degeneracy structure of
the spectrum using the monodromy representation. In particular, we describe with detail the finite, or quasi-exact, eigenfunctions as they are an important part of the discussion on Part \ref{sec:extendeddiscss}.

\section{Definition and basic properties}
\label{sec:biased-ncho-}

 To simplify the notations, we use bold letter for matrices, except when the matrices are considered as elements of the Lie algebra $\mathfrak{sl}_2$.

\begin{dfn}
 Define the $\eta$-shifted non-commutative harmonic oscillator ($\eta$-NCHO) by the Hamiltonian $Q^{(\eta)}=Q^{(\eta)}_{\alpha,\beta}(x,D)$ acting on $L^2(\R)\otimes \C^2$ given by
\[
 Q^{(\eta)}_{\alpha,\beta}(x, D):= \bA \left( -\frac12 \frac{d^2}{d x^2}  + \frac12 x^2 \right) +  \bJ \left( x \frac{d}{d x} + \frac12 \right) + 2 \eta i \det(\bA \pm i\bJ)^{\frac12} \bJ
\]
where $\eta \in \R$.
\end{dfn}

In this paper, we further assume that $\alpha,\beta>0$ and $\alpha\beta >1$ with $\alpha \neq \beta$. Note that under this assumption,  $\det(\bA \pm i\bJ)^{\frac12} = (\det \bA - \rm{pf} \bJ )^{\frac12} = \sqrt{\alpha\beta -1}>0$. The NCHO Hamiltonian $Q$ in \eqref{eq:Qdef} is obtained by taking $\eta=0$.

The operator $Q^{(\eta)}$ is an unbounded, self-adjoint operator with a discrete spectrum bounded below. Indeed, the principal symbol of $Q^{(\eta)}$ is given by
\[
  Q_{\alpha,\beta}^{(\eta)}(x,\xi) =
  \begin{bmatrix}
    \alpha & 0 \\
    0 & \beta 
  \end{bmatrix} \frac{x^2+\xi^2}{2} + i \eta \det(\bA \pm i\bJ)^{\frac12} \bJ x \xi,
\]
therefore, $Q^{(\eta)}$ as a global pseudo-differential operator is positive elliptic and the properties of the spectrum then follow from this fact (see e.g. Theorem 3.3.13 and Remark 3.3.16 of \cite{P2010S}). We note that in the case of the usual NCHO ($\eta=0$), the operator
$Q^{(0)}$ is global positive elliptic. 

The $\eta$-NCHO has a parity symmetry, that is, the Hamiltonian $Q^{(\eta)}$ commutes with the involution $\bI_2 \mathcal{P}$, where $\bI_2$ the
$2\times2$-identity matrix and $\mathcal{P}$ is the parity operator defined by
\[
  (\mathcal{P} f)(x) = f(-x)
\]
for $f \in L^2(\R)$.

\begin{rem}
  We refer to $\eta$-NCHO as ``shifted'' instead of ``asymmetric'' or ``biased'' since the addition term $2 \eta i J$ does
    not break the parity symmetry of the Hamiltonian and to avoid physical connotations. 
\end{rem}

In the study of the properties of the $\eta$-NCHO, we also consider the twisted operator $\bK Q^{(\eta)} \bK$,
where
\[
  \bK =
  \begin{bmatrix}
    0 & 1 \\
    1 & 0
  \end{bmatrix}.
\]
The twisted operator $\bK Q^{(\eta)} \bK$ has the same spectrum as $Q^{(\eta)}$. This fact, in the case of $\eta=0$, was used in
\cite{PW2002,PW2002a} to describe the properties of the spectrum of the NCHO, including its degeneracy. 

\begin{rem}
  Note that $i \bJ = \sigma_y$, where $\sigma_y$ is the usual Pauli matrix and that $\bK= \sigma_x$, however we prefer to keep the notations
  consistent with that of \cite{W2015IMRN} and the literature dealing with the QRM.
\end{rem}

Following \cite{PW2002,PW2002a}, we introduce a particular basis to $L^2(\R)$ to describe a classification of the spectrum of the $\eta$-NCHO.
As is customary, we define annihilation and creation operators in $L^2(\R)$ by
\[
  \psi = (x+ \partial_x)/\sqrt{2}, \qquad \qquad \psi^\dag = (x-\partial_x)/\sqrt{2},
\]
satisfying
\[
  [ \psi, \psi^\dag] = 1.
\]
An orthogonal basis of $L^2(\R)$ is given by the Hermite functions $\phi_n(x) = (\psi^\dag)^n \phi_0 \in L^2(\R)$ for $n\ge0$ with
$\phi_0(x) = e^{-x^2/2}$ (the vacuum vector). In addition, we have
\[
  (\phi_n,\phi_m)_{L^2(\R)} = \sqrt{\pi} \delta_{n,m} n!
\]
where $\delta_{n,m}$ is the Kronecker delta function. We denote by $L^2(\R)_{{\rm fin}}$ the set of finite linear combinations
of $\{\phi_n : n\ge 0\}$ with complex coefficients.

\begin{dfn}
  Denote the spectrum of the $\eta$-NCHO by ${\rm Spec}(\sNCHO{\eta})$. We have
  \begin{align*}
    {\rm Spec}(\sNCHO{\eta}) &= \Sigma_{0} \cup \Sigma_{\infty} \\
                          &= \Sigma^{+}_{0} \cup \Sigma^{-}_{0} \cup \Sigma^{+}_{\infty} \cup \Sigma^{-}_{\infty},
  \end{align*}
  where
  \begin{itemize}
  \item $\Sigma_{0}$ is the set of eigenvalues of ``finite type'': eigenvalues with eigenfunctions consisting of finite sums with respect to the Hermite function basis, that is, elements of $L^2(\R)_{{\rm fin}}\otimes \C^2$,
  \item $\Sigma_{\infty}$ is the set of eigenvalues of ``infinite type'': eigenvalues that are not of finite type. 
  \item $\Sigma_{i}^{+}$ (resp. $\Sigma_{i}^{-}$) for $i \in \{0,\infty\}$ are the eigenvalues having even (resp. odd) eigenfunctions.
  \end{itemize}
\end{dfn}

We note that the definition of the eigenvalues of finite type does not depend on the choice of basis.
In fact, in Section \ref{sec:multiplicity} we give an alternative description of the eigenvalues of $\Sigma_0$ in terms of the solution of Heun ODE associated to the $\eta$-NCHO. To simplify the discussion, we will say that eigenfunctions corresponding to eigenvalues of finite type are quasi-exact and call them quasi-exact solutions as in the case of other physical models.

\begin{rem}
  In this paper we limit ourselves to  the case $\alpha,\beta>0$ and $\alpha\beta >1$ with $\alpha \neq \beta$ but other cases of the parameters may be considered
  resulting in different spectral structure. We refer the reader to \cite{P2010S,P2014Milan} (and \cite{PV2013} for the case $\alpha \beta \leq 1$) for the case of the NCHO ($\eta = 0$).
\end{rem}


\section{Eigenvalue problems under the oscillator representation of $\mathfrak{sl}_2(\R)$ }
\label{sec:HeunPicture}

The representation theory of the Lie algebra $\mathfrak{sl}_2$ allows a compact representation of the $\eta$-NCHO and, more importantly, to describe the corresponding Heun picture. The discussion in this section essentially follows \cite{O2001} (see also \cite{W2015IMRN}).

We begin setting the stage by introducing certain bases and isometries for the spaces $L^2(\R)$ and $\bar{\C}[y]$.
Let $\C[y]$ be the vector space of polynomials in the variable $y$ equipped with the Fisher inner product
given by
\[
  (f,g)_F = \sqrt{\pi} f(\partial_y)\overline{g(y)}|_{y=0}
\]
for $f,g \in \C[y]$.  It is also clear that $(y^n,y^m) = \delta_{n,m} \sqrt{\pi} n! $.

Define the linear map
\[
  T_C : L^2(\R)_{{\rm fin}} \to \C[y]
\]
on basis elements by $T_C(\phi_n) = y^n$. It is then immediate to verify that $T_C$ is an isometry and
that \[
  T_C(\psi^\dag \phi) = y T_C(\phi)
\]
for $\phi \in  L^2(\R)_{{\rm fin}}$. The isometry $T_C$ may be extended to an isometry between the
Hilbert spaces $L^2(\R)$ and $\overline{\C[y]}$ and we denote the extension also by $T_C$.

Let us recall the definition of the oscillator representation $\pi$ of $\mathfrak{sl}_2$ on $L^2(\R)$.
The representation $\pi$ is defined by its action 
\[
  \pi(H) = x \frac{d}{d x} + \frac12, \qquad \pi(E) = \frac{x^2}{2}, \qquad \pi(F) = -\frac12 \frac{d^2}{d x^2},
\]
on the standard generators of $\mathfrak{sl}_2$
\[
  H =
  \begin{bmatrix}
    1 & 0 \\
    0 & -1
  \end{bmatrix}, \quad
 E =
 \begin{bmatrix}
   0 & 1 \\
   0 & 0
 \end{bmatrix}, \quad
 F =
 \begin{bmatrix}
   0 & 0 \\
   1 & 0
 \end{bmatrix}.
\]

For completeness, we recall the commutation relations
\[
  [H,E] = 2 E, \qquad [H,F] = -2 F, \qquad  [E,F] = H,
\]
and we refer to any standard reference (e.g. \cite{HT1992}) for the properties of the Lie algebra $\mathfrak{sl}_2$ and
its representations.

Next, consider the representation $(\pi',\C[y])$ of $\mathfrak{sl}_2$ given by
\[
  \pi'(H) = y \frac{d}{d y} + \frac12, \qquad \pi'(E) = \frac{y^2}{2}, \qquad \pi'(F) = -\frac12 \frac{d^2}{d y^2},
\]
we can then verify that $\pi'(\bC X \bC^{-1}) T_C = T_C \pi(X)$ for $X \in \mathfrak{sl}_2$ and where $\bC$ is the Cayley transform
\[
  \bC = \frac1{\sqrt{2}}
    \begin{bmatrix}
      1 & -i \\
      1 & i
    \end{bmatrix}.
\]

The next result gives an equivalent formulation of the eigenvalue problem of $Q^{(\eta)}$ in a single differential equation
in $\bar{\C[y]}$ (cf. Lemma 2.1 of \cite{W2015IMRN}).
To simplify the expressions we introduce some additional notations. Define
\[
  \delta = \frac12 \frac{\alpha+\beta}{\alpha \beta}, \qquad \epsilon = \left| \frac{\alpha - \beta}{\alpha+\beta} \right|,
\]
and the parameter $\kappa$ by the relations
\[
  \ch(\kappa) = \sqrt{ \frac{\alpha \beta}{\alpha \beta - 1}}, \qquad \sh(\kappa) = \frac{1}{\sqrt{\alpha \beta - 1}}, \qquad \nu  = \lambda \delta \ch(\kappa) =\frac{\alpha + \beta}{2 \sqrt{\alpha \beta(\alpha\beta - 1)}} \lambda.
\]

\begin{prop}
  The eigenvalue problem $Q^{(\eta)} \phi = \lambda\phi$ ($\phi \in L^2(\R)$) is equivalent to the differential equation \( \pi'(\mathcal{R}^{(\eta)}) u = 0 \,\, (u \in \overline{\C[y]} )\) with
  \[
    \mathcal{R}^{(\eta)} = \left(2(E-F) - 2 \cth(2 \kappa) H + \frac{\lambda \delta}{\sh(\kappa)}+ \frac{4 \eta}{\sh(2 \kappa)} \right)(H - \lambda \delta \ch(\kappa) +2\eta) + \cth(\kappa) (\lambda \delta \epsilon)^2,
  \]
  by the isometry $T_C : L^2(\R) \to \overline{\C[y]}$. Similarly, the eigenvalue problem $\bK Q^{(\eta)} \bK \phi = \lambda\phi$ ($\phi \in L^2(\R)$)
  is equivalent to the equation $\pi'(\widetilde{\mathcal{R}^{(\eta)}}) u = 0$ with
  \[
    \widetilde{\mathcal{R}^{(\eta)}} = (H - \lambda \delta \ch(\kappa) + 2\eta)\left(2(E-F) - 2 \cth(2 \kappa) H + \frac{\lambda \delta}{\sh(\kappa)}+ \frac{4 \eta}{\sh(2 \kappa)}   \right) + \cth(\kappa) (\lambda \delta \epsilon)^2.
  \]
  \end{prop}
  
\begin{proof}
  The eigenvalue problem of the $\eta$-NCHO is the solution of the differential equation
  \[
    (\bA \pi(E+F) + \bJ \pi(H) + 2\eta i \sqrt{\alpha \beta -1} \bJ - \lambda \bI_2) \varphi_0(x) = 0,
  \]
  simultaneously for $\lambda \in \R$ and  $\varphi_0 \in L^2(\R) \otimes \C^2$. This is equivalent to 
  \[
    \left( \pi(E+F) + \frac{1}{\sqrt{\alpha \beta}} \bJ \pi(H) + \frac{2 \eta i}{\ch(\kappa)} \bJ - \lambda \bA^{-1}\right) \varphi_1(x) = 0.
  \]
  with $\varphi_1 = \bA^{\frac12}\varphi_0$. By using the Cayley transform we obtain the differential equation system
  \[
    \begin{bmatrix}
      \pi(E+F) - \frac{i}{\sqrt{\alpha \beta}} \pi(H) - \lambda \delta + \frac{2 \eta}{\ch(\kappa)} & - \lambda \delta \epsilon \\
      - \lambda \delta \epsilon & \pi(E+F) + \frac{i}{\sqrt{\alpha \beta}} \pi(H) - \lambda \delta - \frac{2 \eta}{\ch(\kappa)}
    \end{bmatrix} \varphi_2(x) = 0,
  \]
  for $\varphi_2 = \bC \varphi_1$. We simplify the expression by defining
  \[
    S_\pm := E + F \pm i \frac{1}{\sqrt{\alpha \beta}} H \in \mathfrak{sl}_2,
  \]
  so that the eigenvalue problem is given by
  \begin{equation}
    \label{eq:deqsystem1}
    \begin{bmatrix}
      \pi(S_-) - \lambda \delta + \frac{2 \eta}{\ch(\kappa)} & - \lambda \delta |\epsilon| \\
      - \lambda \delta |\epsilon| & \pi(S_+) - \lambda \delta - \frac{2 \eta}{\ch(\kappa)}
    \end{bmatrix} \varphi_2(x) = 0.
  \end{equation}

The system \eqref{eq:deqsystem1} is equivalent to the single differential equation given by
\begin{equation}
  \label{eq:proofEigen1}
  \left[ \left(\pi(S_+) - \lambda \delta- \frac{2 \eta}{\ch(\kappa)}\right)\left(\pi(S_-) - \lambda \delta + \frac{2 \eta}{\ch(\kappa)}\right) - (\lambda \delta \epsilon)^2 \right] \varphi(x) = 0,
\end{equation}
for $\varphi \in L^2(\R)$. Let us note that for $\lambda \neq 0$ the complementary component $\phi$ of the solution of
\eqref{eq:deqsystem1} is given by
\[
  \phi(x) =  \frac{1}{\lambda \delta |\epsilon|} \left(\pi(S_+) - \lambda \delta- \frac{2 \eta}{\ch(\kappa)}\right) \varphi(x).
\]

Next, we change the differential equation \eqref{eq:proofEigen1} using an equivalent representation given by an
inner-automorphism of $\pi'$.  First, we define maps 
$T: L^2(\R) \to L^2(\R)$ and $T': \C[y] \to \C[y]$ by
\[
  (Tf)(x) = e^{i (\sh \kappa) x^2/2} \left( \ch \kappa \right)^{\frac14} f(\sqrt{\ch(\kappa)} x), \qquad\text{ and } \qquad  (T'g)(y) = g\left(\sqrt{\frac{\ch \kappa}{i-\sh \kappa}} y\right),
\]
for $f \in L^2(\R)$ and $g \in \C[y]$. It can be shown (see \cite{O2001}) that $T$ and $T'$ preserve the inner products of the
corresponding spaces, and that $T'$ can be extended to an isometry of $\overline{\C[y]}$ that we also denote by $T'$.

Then, we define the representation $\bar{\pi}$ of $\mathfrak{sl}_2$ of $\mathfrak{sl}_2$ in $\C[y]$ by
\[
  \bar{\pi}(X) := \pi'(\mathbf{G} X \mathbf{G}^{-1}),
\]
where $\mathbf{G}$ is given by
\[
  \mathbf{G} =
  \begin{bmatrix}
    \frac{\ch k}{i -\sh \kappa} & 0 \\
    0 & 1
  \end{bmatrix}
  \begin{bmatrix}
    1 & 1 \\
    -1 & 1
  \end{bmatrix}
  \begin{bmatrix}
    1 & i \sh \kappa \\
    0 & 1
  \end{bmatrix}
  \begin{bmatrix}
    \ch \kappa & 0 \\
    0 & 1
  \end{bmatrix},
\]
and we verify (see Cor. 2.7 of \cite{O2001}) that
\begin{align}
  \label{eq:piident}
  \overline{\pi}(S_-) &= (\sch \kappa ) \pi' (H) \\
  \overline{\pi}(S_{+}) &= (\sch \kappa ) \pi' \left( \ch (2\kappa) H - \sh (2 \kappa ) \left( E - F \right)  \right), \nonumber
\end{align}
and then by setting $u = T' T_C T \varphi$ we see that \eqref{eq:proofEigen1} is equivalent to
\begin{equation*}
  \left[ \left(\bar{\pi}(S_+) - \lambda \delta- \frac{2 \eta}{\ch(\kappa)}\right)\left(\bar{\pi}(S_-) - \lambda \delta + \frac{2 \eta}{\ch(\kappa)}\right) - (\lambda \delta \epsilon)^2 \right] u(x) = 0.
\end{equation*}

Finally, defining 
\[
  \mathcal{R}^{(\eta)} = \left(2(E-F) - 2 \cth(2 \kappa) H + \frac{\lambda \delta}{\sh(\kappa)}+ \frac{4 \eta}{\sh(2 \kappa)} \right)(H - \lambda \delta \ch(\kappa) +2\eta) + \cth(\kappa) (\lambda \delta \epsilon)^2,
\]
we see by direct computation that
\begin{equation*}
  \left[ \left(\bar{\pi}(S_+) - \lambda \delta- \frac{2 \eta}{\ch(\kappa)}\right)\left(\bar{\pi}(S_-) - \lambda \delta + \frac{2 \eta}{\ch(\kappa)}\right) - (\lambda \delta \epsilon)^2 \right] u(x) = - (\tah \kappa) \pi' (\mathcal{R}^{(\eta)}),
\end{equation*}
giving the desired result since $\pi'$ and $\bar{\pi}$ are related by an inner automorphism. The second result is obtained in a
completely analogous way.

\end{proof}

For later use, let us write the second order operators $\mathcal{R}^{(0)}$, $\widetilde{\mathcal{R}^{(0)}}$ corresponding to the NCHO, that is,
\begin{align}
  \label{eq:soeNCHO}
    \mathcal{R}^{(0)} = \left(2(E-F) - 2 \cth(2 \kappa) H + \frac{\lambda \delta}{\sh(\kappa)}\right)( H - \lambda \delta \ch(\kappa)) + \cth(\kappa) (\lambda \delta \epsilon)^2
\end{align}
and
\[
  \widetilde{\mathcal{R}^{(0)}} = (H - \lambda \delta \ch(\kappa))\left(2(E-F) - 2 \cth(2 \kappa) H + \frac{\lambda \delta}{\sh(\kappa)}\right) + \cth(\kappa) (\lambda \delta \epsilon)^2.
\]

\section{Heun picture of the $\eta$-NCHO}
\label{sec:compl-pict-eigenv}

To obtain a more detailed understanding of the $\eta$-NCHO, including the multiplicity of the eigenvalues and the confluence
picture, we consider a generalization of the representation $\pi'$ that is controlled by a parameter $a \in \C$ (a non-unitary principal series of $\mathfrak{sl}_2(\R)$).

Concretely, for $a \in \C$ we define the representation $\pi'_a$ of $\mathfrak{sl}_2$ on $y^{a-1} \C[y]$ given on generators by
\[
  \pi_a'(H) = y \frac{d}{d y} + \frac12, \quad \pi_a'(E) = \frac{y^2}{2}, \quad \pi'_a(F) = -\frac12 \frac{d^2}{d y^2} + \frac{(a-1)(a-2)}{2} \cdot \frac{1}{y^2}.
\]
It is clear that $\pi_{a}' \simeq \pi'$ when $a=1,2$.

The Heun picture of the $\eta$-NCHO is obtained from the complex picture of the eigenvalue problem via a representation
$\varpi_a$ of $\mathfrak{sl}_2$ on $\C[z,z^{-1}]$ related to $\pi_{a}'$ by a {\em modified Laplace transform} \cite{W2015IMRN}
(for $a=1$ the idea appeared in \cite{O2001}). For $a \in \C$, The representation  $\varpi_a$ is defined on generators by
\[
  \varpi_a(H) = z \partial_z + \frac12, \qquad \varpi_a(E) = \frac12 z^2 (z \partial_z + a), \qquad \varpi_a(F) = - \frac1{2 z} \partial_z + \frac{a-1}{2 z^2}.
\]

The modified Laplace transform $\mathcal{L}_a : \overline{\C}[y] \to \overline{\C}[z]$ is defined by
\[
  (\mathcal{L}_au)(z) = \int_{0}^\infty u(y z) e^{-\frac{y^2}2} y^{a-1} d y,
\]
for $u \in \overline{\C[y]}$. The action of $\mathcal{L}_a$ in monomials is given by
\[
  y^n \mapsto 2^{\frac{n+a}2 -1} \Gamma\left(\tfrac{n+a}2\right) z^n,
\]
and by defining an inner product in $z$-space such that $\{ z^n | n \in \Z_{\geq 0} \}$ is an orthogonal system
with
\[
  (z^n,z^n)_a = \frac{\sqrt{\pi} n!}{2^{n+a-2} \Gamma(\frac{n+a}{2})^2},
\]
we verify that $\mathcal{L}_a$ is an isometry.

For $a \neq 1,2$, the modified Laplace transform $\mathcal{L}_a$ is an intertwiner between $\pi_a$ and $\varpi_a$, that is, we have
\[
  \mathcal{L}_a \pi'_a(X) = \varpi_a(X) \mathcal{L}_a,
\]
for $X \in \mathfrak{sl}_2$. For $a=1,2$, $\mathcal{L}_a$ is a quasi-intertwiner of $\pi'_a$ and $\varpi_a$ in the sense that
\begin{align*}
  \mathcal{L}_a \pi' (X) &= \varpi_a(X) \mathcal{L}_a , \qquad (\text{ for } X = H,E),  \\
  (\mathcal{L}_1 \pi' (F))(z) &= \varpi_1(F) (\mathcal{L}_1 u(z)) + u'(0)/(2z),   \\
  (\mathcal{L}_2 \pi' (F))(z) &= \varpi_2(F) (\mathcal{L}_2 u(z)) - u(0)/(2z^2).
\end{align*}
For more details and properties we refer to \cite{W2015IMRN} (Lemma 2.2). 

Clearly, since the parity operator $\bm{I}_2 \mathcal{P}$ commutes with the Hamiltonian of the $\eta$-NCHO, the eigenfunctions must be even or odd
functions. The next result (cf. \cite{W2015IMRN} Prop 3.1) shows when restricted to even (resp. odd) functions, the 
case $a=1$ (resp. $a=2$) of the modified Laplace transform gives an equivalence between the eigenproblems with respect to the
two representations. 


\begin{prop}
  \label{prop:equiv}
  The element $\mathcal{R}^{(\eta)} \in \mathcal{U}(\mathfrak{sl}_2)$ satisfies the following equations
  \begin{align*}
    (\mathcal{L}_1 \pi'_1(\mathcal{R}^{(\eta)}) u) (z) &= \varpi_1(\mathcal{R}^{(\eta)})(\mathcal{L}_1 u)(z) + (\nu -2\eta- \tfrac32) u'(0) z^{-1} \\
    (\mathcal{L}_2 \pi'_2(\mathcal{R}^{(\eta)}) u) (z) &= \varpi_2(\mathcal{R}^{(\eta)})(\mathcal{L}_2 u)(z) + (\nu -2\eta - \tfrac32) u(0) z^{-2}.
  \end{align*}
  In particular, if $u(z)$ is an even solution, then $(\pi'(\mathcal{R}^{(\eta)})u)(z)=0$ is equivalent to
  \[
    \varpi_1(\mathcal{R}^{(\eta)})(\mathcal{L}_1u) (z) =0,
  \]
  and if $u(z)$ is an odd solution, $(\pi'(\mathcal{R}^{(\eta)})u)(z)=0$ is equivalent to
  \[
    \varpi_2(\mathcal{R}^{(\eta)})(\mathcal{L}_2u) (z) =0. \qed
  \]
\end{prop}

As mentioned above, the purpose of using the representation $\varpi_a$ is that it allows us to describe the eigenvalue problem of the $\eta$-NCHO by a particular Heun ODE (with respect to the transformed variable $\omega = z^2 \cth(\kappa)$) using the second order elements
$\mathcal{R}^{(\eta)}$ and $\widetilde{\mathcal{R}^{(\eta)}}$.

\begin{prop} \label{prop:eigenProblem}

  Set $t = \cth^2 \kappa = \alpha\beta$ and $\omega = z^2 \cth(\kappa)$.  Define the parameters
  \[
    A:= \frac14\left( -1-2\nu + 2 a \right), \qquad B:= a - \frac12, \qquad C:= \frac14(3 - 2\nu + 2 a) = 1+A, \quad D:= A.
  \]    
  The eigenvalue problem for the $\eta$-NCHO is captured by $\mathcal{R}^{(\eta)}$ as a Heun differential equation by
  \[
    z^{-a+1} \varpi_a(\mathcal{R}^{(\eta)}) z^{a-1} = 4 \tah(\kappa) \omega (\omega -1)(\omega - t) \Lambda^a(\omega, \partial_\omega),
  \]
  and
  \begin{align}
    \label{eq:HeunA1}
   \Lambda^a(\omega, \partial_\omega) = \frac{d^2}{d \omega^2} + \left( \frac{1+A+\eta}{ \omega}  + \frac{A - \eta}{(\omega -1)} + \frac{A+B+1-C-D + \eta}{(\omega - t)}  \right) \frac{d}{d \omega}  + \frac{ (A+\eta)B \omega - \bar{q}_a}{\omega (\omega -1)(\omega - t)} 
  \end{align}
  where
  \begin{align}
    \label{eq:accesory}
    \bar{q}_a = \left[ - (a-\frac12 - \nu)^2 + (2\eta)^2 + (\epsilon \nu)^2 \right] (t - 1) - 2 (a-\frac12)(a - \frac12 - \nu+2\eta).
  \end{align}
  Similarly, we see that the Heun picture for the eigenvalue problem $\bm{K} Q^{(\eta)} \bm{K} \phi = \lambda\phi$ is given by 
  \[
    z^{-a+1} \varpi_a(\widetilde{\mathcal{R}^{(\eta)}}) z^{a-1} = 4 \tah(\kappa) \omega (\omega -1)(\omega - t) \bar{\Lambda}^a(\omega, \partial_\omega).
  \]
  with
  \begin{align}
    \label{eq:HeunA2}
    \bar{\Lambda}^a(\omega, \partial_\omega) = \frac{d^2}{d \omega^2} + \left( \frac{A+\eta}{ \omega}  + \frac{A+1 - \eta}{(\omega -1)} + \frac{A+B+2-C-D + \eta }{(\omega - t)}  \right) \frac{d}{d \omega}  + \frac{ (A+1+\eta)B \omega - \bar{q}_a}{\omega (\omega -1)(\omega - t)}. 
  \end{align}
\end{prop}

Note that the accessory parameter is $\bar{q}_a$ is the same for both $\mathcal{R}^{(\eta)}$ and $\widetilde{\mathcal{R}^{(\eta)}}$.

\begin{proof}
  The case of $\mathcal{R}^{(\eta)}$ the proof can easily be verified by direct computation using the identities $z \partial_z = 2 \omega \partial_\omega$. For the case of
  $\tilde{\mathcal{R}^{(\eta)}}$, instead of direct computation we can verify that
  \begin{align*}
    z^{-a+1} \varpi_{a}(\mathcal{R}^{(\eta)})& z^{a-1} = z^{-a+1} \varpi_{a}(\mathcal{R}^{(0}) z^{a-1} + \frac{4\eta}{\sh(2\kappa)} \left( z \partial_z + a - \tfrac12 - \nu + 2\eta \right) \\
    &+ 2\eta \left( (z^2+z^{-2} - 2\cth(\kappa)) (z\partial_z + a - \tfrac12) + (z^2-z^{-2})(a-\tfrac12) + \frac{2 \nu}{\sh(2\kappa)} \right),
  \end{align*}
  where $\mathcal{R}^{(0)}$ corresponds to the NCHO (cf. \eqref{eq:soeNCHO}, and
  \begin{align*}
    z^{-a+1} \varpi_{a}(\mathcal{R}^{(\eta)}) z^{a-1} - z^{-a+1} \varpi_{a}(\mathcal{R}^{(0)}) z^{a-1} = z^{-a+1} \varpi_{a}(\widetilde{\mathcal{R}^{(\eta)}}) z^{a-1} - z^{-a+1} \varpi_{a}(\widetilde{\mathcal{R}^{(0)}}) z^{a-1}.
  \end{align*}
  Therefore, we see (cf. proof of Proposition 2.5 of \cite{W2015IMRN}) that
  \begin{align*}
    &z^{-a+1} \varpi_{a}(\mathcal{R}^{(\eta)}) z^{a-1} - z^{-a+1} \varpi_{a}(\widetilde{\mathcal{R}^{(\eta)}}) z^{a-1} \\
    &\qquad = 4 \tah(\kappa) \omega (\omega-1)(\omega-t) \left( \left\{ -\frac1{\omega} + \frac{1}{\omega-1} + \frac{1}{\omega-t} \right\} + \frac{\omega (a-\tfrac12)}{\omega(\omega-1)(\omega-t)} \right),
  \end{align*}
  and the result follows immediately.
\end{proof}


Finding solutions of the Heun ODE equations (for any parity) in a simply connected domain $\Omega \subset \C$ such that $0,1 \in \Omega$ and $\alpha \beta \not\in \Omega$ is equivalent to finding solutions of the eigenvalue problem for the $\eta$-NCHO. This result follows from the discussion of Section 3.2
of \cite{W2015IMRN} without modifications. 

\begin{prop}
  \label{prop:equiv2}

  For even and odd solutions of the eigenvalue problem of the $\eta$-NCHO, there are linear bijections
  \begin{align*}
    \{ \varphi \in L^2(\R)\otimes \C^2 \, | \, \sNCHO{\eta} \varphi = \lambda \varphi, \varphi(-x) = \varphi(x) \} &\simeq \{ f \in \mathcal{O}(\Omega) | H_\lambda^{+} f = 0 \}, \\
    \{ \varphi \in L^2(\R)\otimes \C^2 \, | \, \sNCHO{\eta} \varphi = \lambda \varphi, \varphi(-x) = -\varphi(x) \} &\simeq \{ f \in \mathcal{O}(\Omega) | H_\lambda^{-} f = 0 \}, 
  \end{align*}
  where $\mathcal{O}(\Omega)$ is the set of holomorphic solutions in a simply connected domain $\Omega \subset \C$ with $1,0 \in \Omega$ but $\alpha \beta \notin \Omega$, and
  where the Heun ODE operators $H_\lambda^{\pm}$ are given by
  \[
    H_{\lambda}^{+}(\omega,\partial_\omega) := \Lambda^1(\omega,\partial_\omega), \qquad H_{\lambda}^-(\omega,\partial_\omega) = \Lambda^2(\omega,\partial_\omega),
  \]
  with
  \[
    \lambda  = \frac{2 \sqrt{\alpha \beta(\alpha\beta - 1)}}{\alpha + \beta} \nu.  
  \]\qed
\end{prop}

In particular, Proposition \ref{prop:equiv2} shows that the multiplicity of the eigenvalues of the $\eta$-NCHO is bounded above by $4$, with each parity having at most multiplicity $2$. In the discussion of the eigenvalues of finite type in the next section we reduce the bound to $3$ by showing that there are no degeneracies involving two  solutions of finite type of different parity. 

We remark that for $a \neq 1,2$ the representations $\pi'_a$ does not capture the $\eta$-NCHO directly via the Heun ODE $\pi'(\mathcal{R}^{(\eta)})f=0$, rather
it gives a  generalization that we may call the $\eta$-NCHO of $a$-type, or $\eta$-NCHO of representation type $a$. We note that for any $a$
and $\eta$, the $\eta$-NCHO of $a$-type is also parity preserving. The introduction of the parameter is crucial for the study of the confluence picture in Part \ref{sec:fromsNCHOtoAQRM}. When possible we will consider the results in the full generality instead of limiting to the cases $a=1,2$.

Let us finish the section by giving the standard form of the Heun operator \eqref{eq:HeunA1} that is used in the analysis of the confluent process in Section \ref{sec:cprod}. We define shifted parameters
\begin{gather*}
  \bar{A} = A + \eta, \qquad  \bar{B} = B, \qquad  \bar{C} = 1+A +\eta = C+\eta \\
  \bar{D} = A -\eta = D-\eta, \qquad  \bar{F} = \bar{A}+\bar{B}+1-\bar{C}-\bar{D}
\end{gather*}
so that the standard form of \eqref{eq:HeunA1} is 
\begin{equation}
  \label{eq:heunStd}
    \Lambda^a(\omega, \partial_\omega) = \frac{d^2}{d \omega^2} + \left( \frac{\bar{C}}{ \omega}  + \frac{\bar{D}}{(\omega -1)} + \frac{\bar{F}}{(\omega - t)}  \right) \frac{d}{d \omega}  + \frac{ \bar{A}\bar{B} \omega - \bar{q}_a}{\omega (\omega -1)(\omega - t)},
\end{equation}
and the local properties are given by the Riemann scheme (see e.g. \cite{SL2000} and Appendix \ref{sec:HeunODE})
\[
  \begin{pmatrix}
    0 & 1 & t & \infty &; z \\
    0 & 0 & 0 & a-\frac{1}{2}&; q_a \\
    \frac{1+2\nu - 2a - 4 \eta}4 & \frac{5+2\nu - 2a +4 \eta}4 & \frac{5-2\nu-2a - 4\eta}4 & \frac{-1 -2 \nu + 2a+ 4 \eta}4 &
  \end{pmatrix}.
\]

\section{Quasi-exact eigenfunctions}
\label{sec:finite}

In this section we describe the quasi-exact eigenfunctions of the $\eta$-NCHO, that is, the eigenfunctions corresponding to eigenvalues $\lambda \in \Sigma_0$. By the discussion in Section \ref{sec:compl-pict-eigenv}, these solutions correspond to even or odd polynomial solutions of $H_\lambda^{+} f = 0$ with respect to the variable $z$.
More generally, we consider the general case $a \in \C$ since for any value of the representation parameter $a \in \C$, the corresponding $\eta$-NCHO of $a$-type is parity preserving.

For $L \in \Z_{\ge0}$, we write
\[
  p(z) = \sum_{\substack{0 \leq n \leq L \\ n \equiv L \pmod{2}}} r_n z^{n},
\]
for a solution of degree $L$, with $r_L \neq  0 $ and with the parity being naturally determined by the degree $L$. For $N \in \Z_{\geq0}$, 
note that the number of nonzero coefficients for both even $L=2 N$ and odd $L=2N+1$ cases is exactly $N$. The following result characterizes the eigenvalues $\lambda \in \Sigma_0$.

\begin{thm}
  \label{thm:finite}
  Let $L \in \Z_{\ge0}$, $a \neq -L$  and $\rho \in \{0,1\}$ such that $L \equiv \rho \pmod{2}$. There is an eigenvalue $\lambda \in \Sigma_0$ of the form
  \begin{equation}
    \label{eq:finiteEigen}
    \lambda = \frac{2 \sqrt{\alpha \beta (\alpha \beta -1)}}{\alpha+\beta} \left(L + \tfrac12 + 2 \eta\right)
  \end{equation}
  if the parameters $(\alpha,\beta,\eta)$ satisfy the constraint condition
  \begin{equation}
    \label{eq:ccond}
    \det(\bM_L^{(a,\rho)}(\alpha,\beta,\eta)) = 0.
  \end{equation}  
  Here, the matrix $\bM_L^{(a,\rho)}(\alpha,\beta,\eta)$ is a tridiagonal $\left(\frac{L-\rho}2 \right) \times \left(\frac{L-\rho}2 \right)$ given by
  \begin{equation}
    \label{eq:tridiag}
    \bM_L^{(a,\rho)}(\alpha,\beta,\eta) = \Tridiag{c^{(L)}_{L-2i}}{d^{(L)}_{L-2i-2}}{f^{(L)}_{L-2i-2}}{1\leq i \leq \frac{L-\rho}2},
  \end{equation}
  with entries
  \begin{align*}
    c_i^{(L)} (= c_i^{(L)}(\alpha,\beta,\eta)) &:=   \frac{\epsilon^2 (2L + 1 + 4 \eta)^2}{2 \sh(2\kappa)} + (i - L) \left(\frac{2L +  1 + 8\eta}{\sh(2\kappa)} - \cth(2\kappa) (2 i  +1)\right), \\
    d_i^{(L)} (= d_i^{(L)}(a)) &:= (i - a + 3 ) (i +2 - L)\\
    f_i^{(L)} (= f_i^{(L)}(a)) &:= (i + a) (i -L ).
  \end{align*}

  In addition, the condition \eqref{eq:ccond} determines uniquely a polynomial eigenfunction of $\eta$-NCHO with degree $L$ associated
  to an eigenvalue $\lambda$ of the form \eqref{eq:finiteEigen}.

  Conversely, eigenvalues $\lambda \in \Sigma_0$ are of the form \eqref{eq:finiteEigen} for $L = L(\alpha,\beta) \in \Z_{\geq 0}$ and have multiplicity one.
  Moreover, we have \( \Sigma_0^+ \cap \Sigma_0^- = \emptyset\). 
\end{thm}

\begin{rem}

  Let us remark the effect of the parameters $L$ and $a$ in the constraint condition \eqref{eq:ccond}.
  The parameter $L$ determines the dimension of the tridiagonal matrix in \eqref{eq:tridiag} and therefore, we
  have
  \[
    \det(\bM_L^{(a,\rho)}(\alpha,\beta,\eta)) = f_L(a,\alpha,\beta,\eta)
  \]
  for some algebraic function $f_L(x,y,z,w)$ with rational coefficients (cf. Section \ref{sec:cpoly}). On the other hand,
  for fixed $L$, for any $a,a'>0$ the corresponding constraint conditions are given by
  \[
    f_L(a,\alpha,\beta,\eta) = 0, \text{ and }  f_L(a',\alpha,\beta,\eta) = 0,
  \]
  with the same $f_L(x,y,z,w)$ but note that the solutions differ in general. The fact that the form of the constraint condition,
  described by the function $f_L(x,y,z,w)$ does not change withe the parameter $a$ is used in
  Section \ref{sec:cpoly} to study the limit of quasi-exact solutions of the $\eta$-NCHO under the iso-parallel confluence process.
\end{rem}

\begin{proof}
 First, we consider the action on a monomial
  \begin{align*}
    \varpi_a(\mathcal{R}^{(\eta)})z^{m} =  &\left( (z^2+z^{-2})z \partial_z +  a(z^2-z^2) + \frac1{z^2} -  \cth(2 \kappa) (2z \partial_z + 1) + \frac{\lambda \delta}{\sh(\kappa)} + \frac{4 \eta}{\sh(2\kappa)} \right)  \\
       & \times \left( m + \frac12 - \lambda\delta \ch(\kappa) + 2 \eta \right) z^{m} + \cth(\kappa) (\lambda\delta\epsilon)^2 z^{m} \\
    =& (m + \frac12 - \nu  + 2 \eta) (m + a) z^{m+2} \\
       &+ \left[  \frac{2\epsilon^2 \nu^2}{\sh(2\kappa)} + (m + \frac12 - \nu + 2\eta) \left(\frac{2 \nu+4\eta}{\sh(2\kappa)} - \cth(2\kappa) (2 m +1)\right)\right] z^{m} \\
  &+ (m + \frac12 - \nu + 2\eta)(m  - a + 1) z^{m-2},
\end{align*}
therefore, for a solution $p(z)$ we have that $\varpi_a(\mathcal{R}^{(\eta)}) p(z) =0$ is given by
\begin{align}
  \label{eq:finiteSol}
  \varpi_a(\mathcal{R}^{(\eta)}) p(z) =& \sum_{\substack{1 \leq m \leq L+2 \\ n \equiv L \pmod{2}}} r_{m-2} (m - 2 + \frac12 - \nu + 2 \eta)(m-2 + a)z^{m} \\
       & + \sum_{\substack{0 \leq m \leq L \\ n \equiv L \pmod{2}}}  r_m \left[  \frac{2\epsilon^2 \nu^2}{\sh(2\kappa)} + (m + \frac12 - \nu + 2\eta) \left(\frac{2 \nu+4 \eta}{\sh(2\kappa)} - \cth(2\kappa) (2 m +1)\right)\right] z^{m} \nonumber \\
       & + \sum_{\substack{0 \leq m \leq L-2 \\ m \equiv L \pmod{2}}} r_{m+2} (m  + 2 +  \frac12 - \nu + 2\eta)(m -a + 3) z^{m} = 0. \nonumber
\end{align}

Immediately, by considering the coefficients of degree $L+2$ in \eqref{eq:finiteSol}, we see that the condition imposed is
\begin{equation}
  \label{eq:reccond1}
  r_L (L + \frac12 - \nu + 2\eta) (L +a) = 0
\end{equation}
so that
\[
  \nu = L + \frac12 + 2\eta,
\]
since $a \neq -L$ (see Remark \ref{rem:amL} for the case $a= -L$).  Note that in addition to \eqref{eq:reccond1}, the coefficients of degree $L$ in \eqref{eq:finiteSol} impose the condition
\[
  \frac{\epsilon^2 (2L+ 1 + 4 \eta)^2}{2 \sh(2\kappa)} r_L   - 2 (L+a - 2) r_{L-2}  = 0,
\]
that determines the coefficient $r_L$.

Let us now consider the constraint equation for the existence of the solutions. We assume that $\nu = L + \frac12 + 2\eta $, then
the equation
\[
  \varpi_a(\mathcal{R}^{(\eta)}) p(z) = 0
\]
determines a three-term recurrence relation for the coefficients $r_m$
\begin{align}
  \label{eq:recRel2}
  & r_{m+2} (m -a + 3)(m+2 - L) \\
  & \qquad + r_m \left[  \frac{\epsilon^2 (2L+ 1 + 4 \eta)^2}{2 \sh(2\kappa)} + (m - L) \left( \frac{2L + 1 + 8\eta}{\sh(2\kappa)} - \cth(2\kappa) (2m +1)\right)\right]  \nonumber \\
  &\qquad \qquad + r_{m-2} (m+a - 2) (m-2 -L) = 0 \nonumber
\end{align}
for $m=\rho,\rho+2,\cdots,L-4$.

The constraint condition is given by the coefficients in \eqref{eq:reccond1} for $m=L-2$, in this case, the coefficient of $r_L$ vanishes and we obtain
\begin{align*}
  &r_{L-2} \left[  \frac{\epsilon^2 (2L+ 1 + 4 \eta)^2}{2 \sh(2\kappa)} - 2 \left( \frac{2L + 1 + 8\eta}{\sh(2\kappa)} - \cth(2\kappa) (2L-3)\right)\right] \\
  &\qquad  -4 r_{L-4} (L + a - 4) = 0,
\end{align*}
Therefore, the constraint condition is given by 
\begin{equation}
  \label{eq:constcond}
  \det\left(\bM_L^{(a,\rho)}(\alpha,\beta,\eta)\right) = 0,
\end{equation}
where $\bM_L^{(a,\rho)}(\alpha,\beta,\eta)$ is the desired tridiagonal $\left(\frac{L-\rho}2 \right) \times \left(\frac{L-\rho}2 \right)$-matrix.

The last statement of the proof follows from the fact that polynomial solutions corresponding to $\lambda \in \Sigma_0$ of the
form \eqref{eq:finiteEigen} with $L$ of different parity have the corresponding parity, therefore there no is
degeneracy between two eigenvalues of finite type of different parity.
\end{proof}

\begin{rem}
  \label{rem:amL}
  Let us suppose that the condition \eqref{eq:reccond1} is satisfied with $a=-L$, then it is possible to have a solution of degree $L$ for
  arbitrary $\nu$. Concretely, since $a = -L$ makes the coefficient of degree $L+2$ vanish in \ref{eq:finiteSol} and the condition for the
  coefficient of degree $L$ is given by
  \begin{align*}
    -2 r_{L-2} &(L - 2 + \frac12 - \nu + 2 \eta) \\
   & + r_L \left[  \frac{2\epsilon^2 \nu^2}{\sh(2\kappa)} + (L + \frac12 - \nu + 2\eta) \left(\frac{2 \nu+4 \eta}{\sh(2\kappa)} - \cth(2\kappa) (2 L +1)\right)\right] =0,
  \end{align*}
  the resulting the constraint relation
  \begin{equation}
    \label{eq:constamL}
    \det M_{L}^{(-L,\rho)}(\alpha,\beta;\nu)=0,
  \end{equation}
  depends on the value of $\nu$. Therefore, a triplet $(\alpha,\beta,\nu)$ of parameters that satisfy the condition \eqref{eq:constamL} determine
  possible values for $\nu$ (and therefore the eigenvalue $\lambda$) for a finite solution in this case.

  A more detailed analysis of the constraint condition \eqref{eq:constamL} and its solutions is outside of the scope of this paper, but
  we mention a special case that is of particular interest, the case when $\nu = L-2 + \frac12 + 2\eta$. In this case the solution degenerates
  into a $L-2$ degree polynomial solution and the constraint relation reduces to
  \[
    \det \bM_{L-2}^{(-L,\rho)}(\alpha,\beta,\eta)=0,
  \]
  and is therefore one of the cases already covered in Theorem \ref{thm:finite}. After shifting $L$ appropriately,
  this solution may be considered an element of a subrepresentation of $(\varpi_{-L}, z^{L+2}\C[z^2, z^{-2}])$,
  where the representation gives the Langlands quotient $(\varpi_{-L}, z^{(L+2)}\C[z^2, z^{-2}]/z^L\C[z^{-2}])$
  (cf. Lemma 2.3 of \cite{W2015IMRN}\footnote{Note that there is a typo on Lemma 2.3 of \cite{W2015IMRN}: in the
    second equivalence, $a$ should be $2-a$. However, the correct expression should be evident to the reader from
    Remark 2.4 or Remark 2.5 below Lemma 2.3.}) which is equivalent $(\pi’_{-L},y^{L+2}\C[y^2])$ by the intertwiner $\mathcal{L}_{-L}$.
\end{rem}

We obtain as an immediate corollary the expression for the solutions of $\varpi_a(\mathcal{R}^{(\eta)}) p(z)=0$ corresponding to $\lambda \in \Sigma_0$ and
the recurrence relation for the coefficients.

\begin{cor}
  \label{cor:finsol}
  With the notation of Theorem \ref{thm:finite}, if there is an eigenvalue $\lambda\in \Sigma_0$ of $\eta$-NCHO given by
  \[
    \lambda = \frac{2 \sqrt{\alpha \beta (\alpha \beta -1)}}{\alpha+\beta} \left(L + \tfrac12 + 2 \eta\right),
  \]
  the polynomial solution $p(z)$ of $\varpi_a(\mathcal{R}^{(\eta)}) p(z)=0$ of degree $L$ of the form
  \[
    p(z) = \sum_{\substack{0 \leq n \leq L \\ n \equiv L \pmod{2}}} r_n z^{n},
  \]
  has coefficients determined by the three-term recurrence relation
  \begin{align*}
    &(m-a+1)(m-L) r_m  =\\
    & \qquad - \Bigg[ \frac{\epsilon^2 (2L+1+4\eta)^2}{2 \sh(2\kappa)} + (m-2-L)\left( \frac{2L+1+8\eta}{\sh(2\kappa)} - \cth(2 \kappa) \left( 2(m-2+1) \right) \right) \Bigg] r_{m-2} \\
         & \qquad \qquad \qquad \qquad \qquad \qquad + (m+a-4)(m-4-L) r_{m-4} ,
  \end{align*}
  for $\rho \leq m \leq L-2$ with $r_{\rho}=1$ and $r_{\rho-2}=0$ where $\rho \in \{0,1\}$ according to the parity of $L$. In addition, the coefficient
  $r_L$ is determined by
  \[
    r_L =  \frac{2 (L-2+a)\sh(2\kappa)}{2\epsilon^2 (L+\frac{1}{2}+2\eta)^2}   r_{L-2},
  \]
  the parameters $\alpha,\beta$ are subject to the constraint condition \eqref{eq:constcond} of Theorem \ref{thm:finite}.
\end{cor}

We summarize the result of Proposition \ref{thm:finite} as
\[
  \Sigma_0 \subset \left\{ \frac{2 \sqrt{\alpha \beta (\alpha \beta -1)}}{\alpha+\beta} \left(L + \tfrac12 + 2 \eta\right)  : L \in \Z_{\geq 0} \right\}
\]
and more importantly, since the parameter $L$ of the $\lambda \in \Sigma_0$ of the form \eqref{eq:finiteEigen} determine the degree, and thus the
parity of the eigenfunctions, we have shown that there is no degeneracy between even and odd eigenvalues of finite type, that is
\[
  \Sigma_0^{+} \cap \Sigma_0^{-} = \emptyset.
\]

Moreover, if the parameters $\alpha,\beta$ satisfy the constraint equation, the coefficients are determined uniquely (up to constant multiplication) and therefore the eigenvalue $\lambda \in \Sigma_0$ has multiplicity one in $\Sigma_0$ . Thus the multiplicity of any given eigenvalue
$\lambda$ of $\sNCHO{\eta}$ is bounded above by three and this situation can only be achieved in the case
\[
  \lambda \in  \Sigma_{0}^{\pm} \cap \Sigma_{\infty}^{\pm} \cap \Sigma_{\infty}^{\mp}.
\]
In the next section, we show that the multiplicity of the eigenvalues is actually bounded above by $2$ as in the case of the NCHO.

\begin{rem}
  We may use the Heun ODE \eqref{eq:HeunA2} to obtain, for an eigenvalue $\lambda \in \Sigma_0$ of the form
  \eqref{eq:finiteEigen}, a constraint relation of the form
  \[
    \det\left(\widetilde{\bM_L}^{(a,\rho)}(\alpha,\beta,\eta)\right) = 0,
  \]
  for some tridiagonal matrix $\widetilde{\bM_L}^{(a,\rho)}(\alpha,\beta,\eta)$. The components of the matrix
  $\widetilde{\bM_L}^{(a,\rho)}(\alpha,\beta,\eta)$ are given by
  \begin{align*}
    \widetilde{c_i}^{(L)} &= c_i^{(L)} \\
    \widetilde{d_i}^{(L)} &= (i+a)(i+ 2 -L) \\
    \widetilde{f_i}^{(L)} &= (i-a+3)(i - L),
  \end{align*}
  such that $d_i^{(L)} f_i^{(L)} =  \widetilde{d_i}^{(L)} \widetilde{f_i}^{(L)}$ and thus the determinants coincide by basic
  properties of the continuants. This is expected since the spectrum of $Q^{(\eta)}$ and $\bm{K} Q^{(\eta)} \bm{K}$ is equal. However, we
  note that the polynomial solutions for this case are of degree $L-2$.
\end{rem}

\begin{rem}
  The results of Theorem \ref{thm:finite} suggests that, in analogy to the AQRM, we can consider the eigenvalues
  of the form given in \eqref{eq:finiteEigen} to be ``exceptional eigenvalues'', that is, the remains of the
  eigenvalues of the (shifted) quantum harmonic oscillator. Moreover, since all eigenvalues of $\Sigma_0$ are of that
  form, the solutions corresponding to eigenvalues $\Sigma_0$ may be considered analogous to the Juddian eigenvalues
  of the AQRM. In the case of the QRM, the existence of non-Juddian solutions, that is, solutions corresponding
  to exceptional eigenvalues of ``infinite type'' was not immediately recognized. As we see in the next
  section, for the case of $\eta \in \frac12\Z$, we can clarify the situation for the $\eta$-NCHO.
\end{rem}

\subsection{Curves determined by constraint conditions}
\label{sec:curvespoly}

In this section we give example of the curves determined by the constraint relation for the $\eta$-NCHO, that is, 
\[
  \det(\bM_L^{(a,\rho)}(\alpha,\beta,\eta)) = 0.
\]

To simplify, we assume that $a=1,2$ according to the parity of $L$ and we consider as axis
\[
  x = \frac{\alpha}\beta, \qquad y = \frac{1}{\sqrt{\alpha \beta}}.
\]
such that $\alpha,\beta >0$ and $\alpha \beta >1$ gives $0< x$ and $0<y\leq 1$.  Note that these $x$ and $y$ are the essential variables in the
constraint condition for ensuring the existence of the quasi-exact solution of the $\eta$-NCHO. With these variables, we have
\begin{align*}
  c_i^{(L)}  =&  y \Bigg[ \left(\frac{x-1}{x+1}\right)^2\frac{ (2 L + 1 + 4 \eta)^2(y^{-2} -1)}{4} \\
  & \quad \quad \quad   + (i - L) \left\{(L-i + 4\eta)y^{-2} -  (L+1 + i + 4\eta)\right\} \Bigg].
\end{align*}

For $\eta=\frac12$, the curves are given in the following figures for different values of $L$.

\begin{figure}[h]
  \centering
  \subfloat[$L=3$]{
    \includegraphics[height=4.5cm]{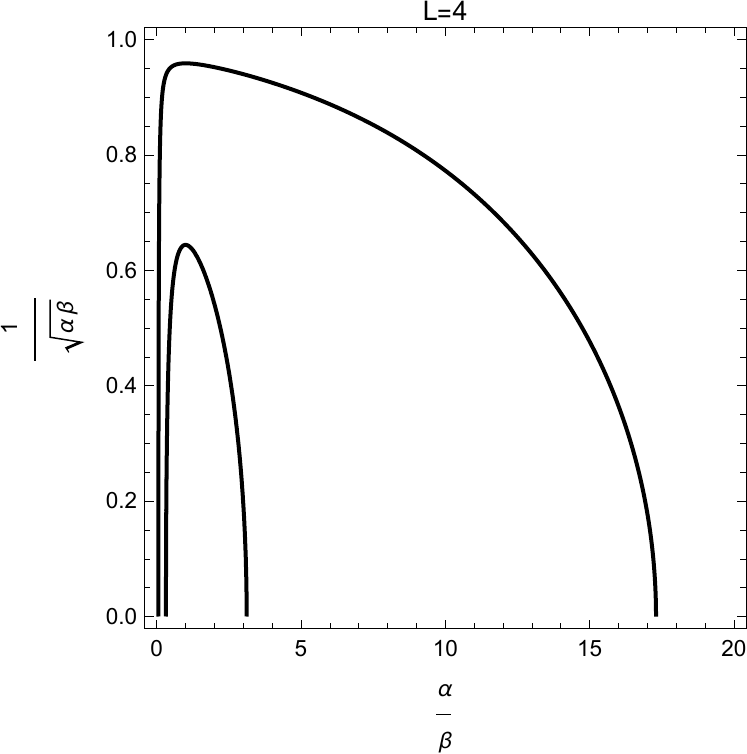}}
  ~ \qquad \qquad
  \subfloat[$L=4$]{
    \includegraphics[height=4.5cm]{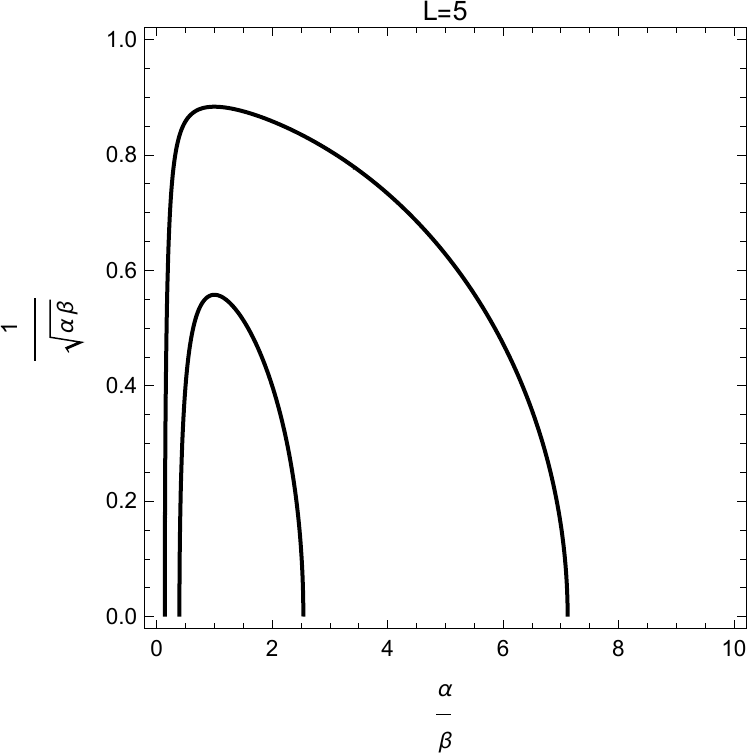}}
  \caption{Curves determined by the constraint condition for $\eta=\frac12$}
  \label{fig:curve1}
\end{figure}

\begin{figure}[h]
  \centering
  \subfloat[$L=5$]{
    \includegraphics[height=4.5cm]{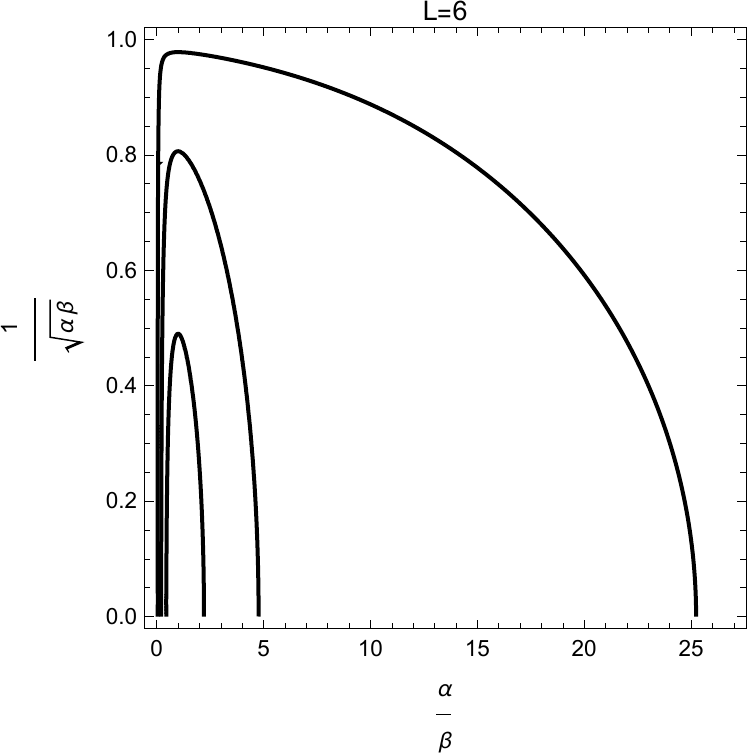}}
  ~ \qquad \qquad
  \subfloat[$L=6$]{
    \includegraphics[height=4.5cm]{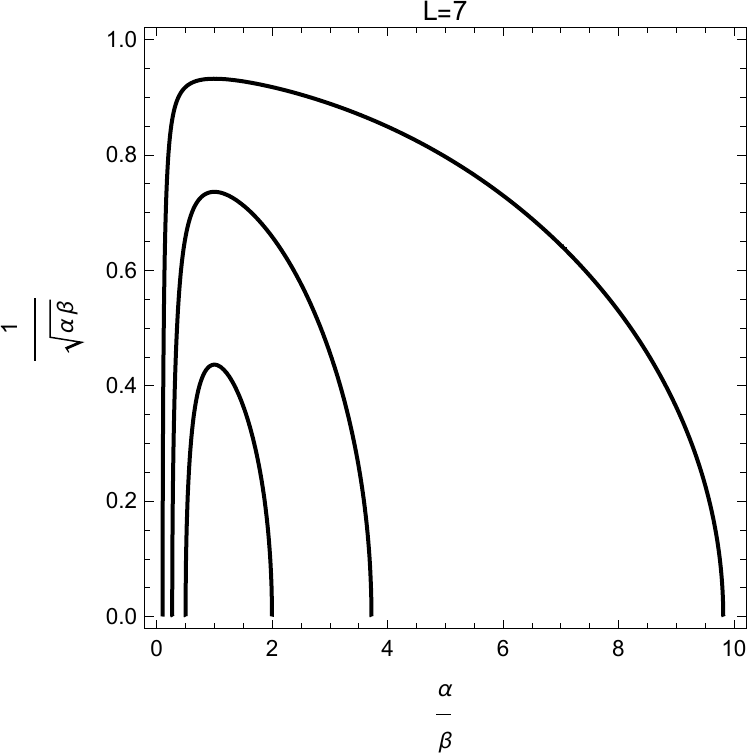}}
  \caption{Curves determined by the constraint condition for $\eta=\frac12$}
  \label{fig:curve4}
\end{figure}

\begin{figure}[h]
  \centering
  \subfloat[$L=10$]{
    \includegraphics[height=4.5cm]{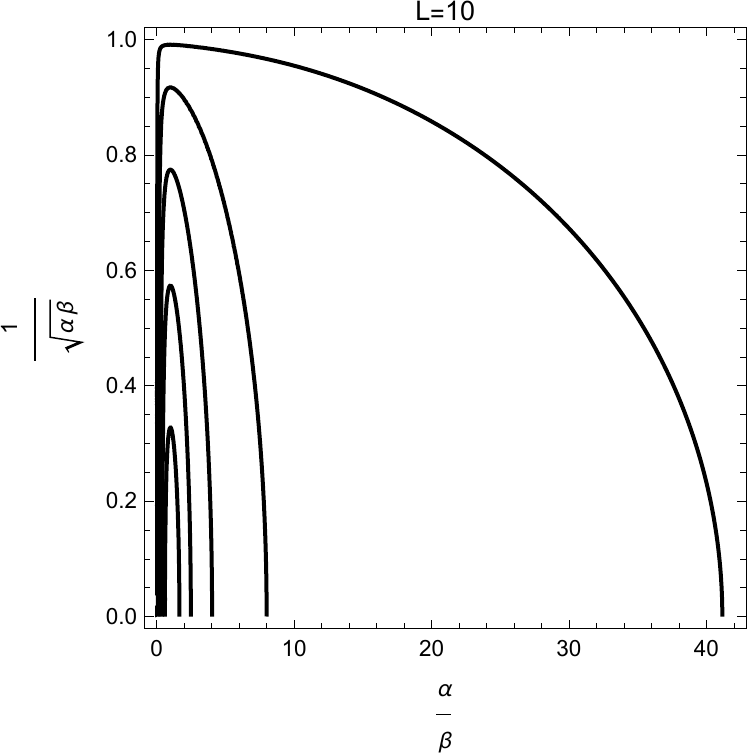}}
  ~ \qquad \qquad
  \subfloat[$L=10$ (detail)]{
    \includegraphics[height=4.5cm]{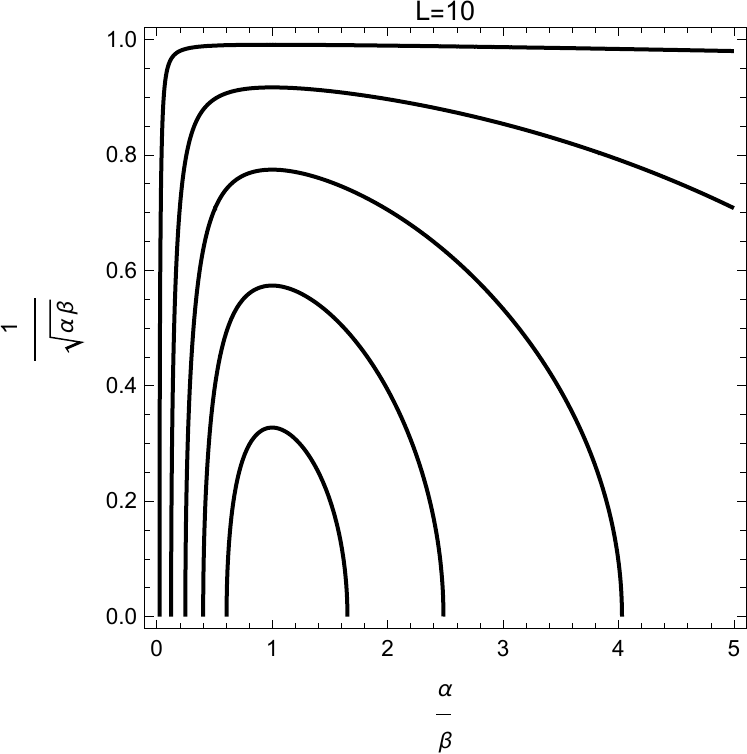}}
  \caption{Curves determined by the constraint condition for $\eta=\frac12$}
  \label{fig:curve3}
\end{figure}

\begin{ex}

  The rational functions in $x, y$ for the case $L=3$ is given by
  \begin{align*}
    \det(\bM_3^{(a,1)}(x,y,\eta)) &= -\frac1{4 (x+1)^2 y}(9 x^2 y^2-33 x^2+16 \eta ^2 (x-1)^2 \left(y^2-1\right) \\
                             &+8 \eta  (x (3 x-22)+3) \left(y^2-1\right)-178 x y^2+130 x+9 y^2-33), \\
  \end{align*}
\end{ex}

\section{Spectral structure of $\eta$-NCHO}
\label{sec:multiplicity}

In this section we show that the multiplicity of the eigenvalues of the $\eta$-NCHO is actually bounded above by $2$ and we describe the degeneracies of type $\lambda \in \Sigma_0^{\pm} \cap \Sigma_\infty^{\pm}$ (double-sign corresponds), including the general form of the solutions.
Notably, we show that these type of degeneracies appear only when $\eta \in \frac12 \Z$. The exposition here using monodromy follows \cite{O2001, O2004} and \cite{W2015IMRN}.
In the proofs, the singularity data of the Heun ODE $H_{\lambda}^{\pm}$, summarized in the Riemann scheme, is fundamental.

First, recall the relation
\[
  \lambda = \frac{2 \sqrt{\alpha \beta (\alpha \beta -1)}}{\alpha+\beta} \nu,
\]
between the eigenvalue $\lambda$ and the parameter $\nu$. We define
\[
  \sigma = \sigma(\lambda,\eta) = \frac{1}{4}\left( 2 \nu - 3 - 4\eta \right),
\]
then, for $H_{\lambda}^+$, the Riemann scheme is given by
\begin{equation}
  \label{eq:grsEven}
  \begin{bmatrix}
    0 & 1 & \alpha \beta & \infty  &; \omega \\
    0 & 0 & 0 & \frac12   &; \bar{q}_1 \\
    \sigma + \frac12 & \sigma+\frac32 + 2\eta  & -\sigma-2\eta & -(\sigma+\frac12) &
  \end{bmatrix},
\end{equation}
while for $H_{\lambda}^-$ we have the Riemann scheme
\begin{equation}
  \label{eq:grsOdd}
  \begin{bmatrix}
    0 & 1 & \alpha \beta & \infty  &; \omega \\
    0 & 0 & 0 & \frac32   &; q_2 \\
    \sigma & \sigma+1 + 2\eta  & -\sigma -\frac12-2\eta & -\sigma &
  \end{bmatrix}.
\end{equation}

For eigenvalues in $\Sigma_0$ let us also write explicitly the associated Riemann schemes.
Let $N\in \Z_{\ge0}$. In the even case, we have $\nu= 2 N + \frac12 + 2\eta$ and $\sigma = N-\frac12$, therefore the Riemann scheme for $H_{\lambda}^+$ is given by
\begin{equation}
  \label{eq:grsEven2}
  \begin{bmatrix}
    0 & 1 & \alpha \beta & \infty  &; \omega \\
    0 & 0 & 0 & \frac12   &; \bar{q}_1 \\
    N & N+1+2\eta  & - N + \frac12 -2\eta & -N &
  \end{bmatrix},
\end{equation}
while for the odd case, we have $\nu= 2N +1+ \frac12 + 2\eta$ and $\sigma =N$, therefore the Riemann scheme for $H_{\lambda}^-$ is given by
\begin{equation}
  \label{eq:grsOdd2}
  \begin{bmatrix}
    0 & 1 & \alpha \beta & \infty  &; \omega \\
    0 & 0 & 0 & \frac32   &; \bar{q}_2 \\
    N & N+1 + 2\eta  & -N -\frac12-2\eta & -N &
  \end{bmatrix}.
\end{equation}

With these preparations we present the main result of this section.

\begin{thm}
  The multiplicity of any eigenvalue of the $\eta$-NCHO is bounded above by $2$.
\end{thm}

\begin{proof}
  Let us assume that on the contrary, there is an eigenvalue $\lambda$ of $\eta$-NCHO of multiplicity $3$, that is,
  \[
    \lambda \in  \Sigma_{0}^{\pm} \cap \Sigma_{\infty}^{\pm} \cap \Sigma_{\infty}^{\mp}.
  \]
  It is enough to consider either of the two cases since the complementary one is analogous.
  We consider the case $\lambda \in \Sigma_{0}^{+} \cap \Sigma_{\infty}^{+} \cap \Sigma_{\infty}^{-}$. Since$\lambda \in \Sigma_0^+$, we have $\nu= 2 N+\frac12+2\eta$ for
  some $N \in \Z_{\geq 0}$ and according to the Riemann scheme \eqref{eq:grsEven2} of $H_{\lambda}^+$ we have
  $\sigma=N-\frac12 \in \frac12 \Z$. Let $f_1(\omega)$ and $f_2(\omega)$ be the corresponding even eigenfunctions, a polynomial
  and a holomorphic solution, respectively. In addition, we denote by $\tilde{f}_j(z)$ ($j=1,2$) the
  corresponding solutions with respect to the variable $z$, that is, $\varpi_1(\mathcal{R}^{(\eta)})\tilde{f}(z)=0$.

  Let $u_j = \mathcal{L}_{1}^{-1}\tilde{f}_j$ ($u_j \in \overline{\C[y]}$). We can find a linear combination $u^{+}$ of $u_1$
  (an even polynomial) and $u_2$ such that $\pi'(\mathcal{R}^{(\eta)})u^{+}=0$, that is,  the constant term of $u^+$ is zero.
  In particular, by Proposition \ref{prop:equiv}, we have $\omega_2(\mathcal{R}^{(\eta)})(g^{+})(z)=0$ with $g^+(z) = \mathcal{L}_2 u^{+} $.
  We note that $g^{+}(z) \in \mathcal{O}(\Omega')$, where $\Omega'$ is a connected and simply connected domain with
  $0,\pm(\alpha \beta)^{-\frac14} \in \Omega'$ and $\pm(\alpha \beta)^{\frac14} \notin \Omega'$ and with $g^+(0)=0$. Next, we define
  $g^+(\omega) = z^{-1} g^{+}(z)$ with $\omega = z^2\cth(\kappa)$ and we easily verify that $g^+(\omega) \in \sqrt{\omega}\mathcal{O}(\Omega)$ is a
  solution of
  \[
    H_\lambda^- g^+(\omega) = 0
  \]
  (cf. Proposition \ref{prop:equiv2}) and in particular, we see that $g^+(\omega)$ is holomorphic at $\omega=1$.

  Next, according to the Riemann scheme \eqref{eq:grsOdd} (with $\sigma \in \frac12 \Z_{\ge1} $), $g^+(\omega)$ must have
  exponent $0$ at $\omega=1$. Since we have assumed $\lambda \in \Sigma_\infty^-$, then according to Proposition \ref{prop:equiv2}
  there must be a solution $g^-(\omega)$ of $H_\lambda^-(\omega)=0$ that is holomorphic at $\omega=0,1$ and therefore it has
  exponents $0$ at these points. From this we also see that $g^{+}(\omega)$ must have exponent $\sigma \in \frac12 \Z_{\ge 1}$
  at $\omega=0$. This means that $g^{+}(\omega)$ and $g^{-}(\omega)$ both are linearly independent solutions with exponent
  $0$ at $\omega=1$, but according to the Riemann scheme there must a local solution at $\omega=1$ with exponent $\sigma+1+2\eta$
  and therefore that we have three linear independent local solutions at $\omega=1$. This contradiction completes
  the proof that the multiplicity of $\lambda$ is bounded above by $2$.
\end{proof}

Next, using the  monodromy representation around the singularities (see e.g. \cite{H2020}) we investigate the cases where the multiplicity of eigenvalues $\lambda \in \Sigma_0$ is actually $2$, that is, if $\Sigma_0^\pm \cap \Sigma_\infty^\pm \ne  \emptyset$. Let us denote the monodromy matrices corresponding to the odd case
(i.e. to the Riemann scheme \eqref{eq:grsOdd2}) by $\mathbf{A}_i$ ($i=0,1,2,3$) corresponding to singularities
at $\omega=0,1,\alpha \beta, \infty$ in that order. Similarly, we denote $\mathbf{B}_i$ the monodromy matrices for the even case. In both cases, the product of all the monodromy matrices
is equal to the identity matrix, i.e.
\[
  \mathbf{A}_0 \mathbf{A}_1 \mathbf{A}_2 \mathbf{A}_3 = \mathbf{I}_2  = \mathbf{B}_0 \mathbf{B}_1 \mathbf{B}_2 \mathbf{B}_3.
\]
The eigenvalues of the monodromy matrices are determined by the corresponding Riemann scheme. Concretely, for a singularity
at $\omega \in \{0,1,\alpha \beta ,\infty \}$, the eigenvalues of the corresponding monodromy matrix are given by $\exp(2 i \pi \rho_j(\omega))$ for
$j=1,2$ where $\rho_j(\omega)$ are the two exponents at the singularity $\omega$.

We describe the eigenvalues of the monodromy matrices for the odd case (resp. even case) in Table \ref{tab:eigenMon2}
(resp. Table \ref{tab:eigenMon1}). It is clear that the monodromy is highly dependent the the value of $\eta$. 

\begin{table}[h]
  \centering
  \begin{tabular}{ r | c c c c }
    & $\mathbf{A}_0$ & $\mathbf{A}_1$ & $\mathbf{A}_2$ & $\mathbf{A}_3$  \\
    \hline
    $\epsilon_1$ & 1 & 1 & 1 & -1  \\
    $\epsilon_2$ & 1 & $e^{2 \pi i (2\eta)}$ & $e^{-2 \pi i (\frac12+ 2\eta)}$ & 1
  \end{tabular}
  \caption{Eigenvalues of monodromy matrices for the odd case}
  \label{tab:eigenMon2}
\end{table}

\begin{table}[h]
  \centering
  \begin{tabular}{ r | c c c c }
    & $\mathbf{B}_0$ & $\mathbf{B}_1$ & $\mathbf{B}_2$ & $\mathbf{B}_3$  \\
    \hline
    $\epsilon_1$ & 1 & 1 & 1 & -1  \\
    $\epsilon_2$ & 1 & $e^{2 \pi i (2\eta)}$ & $e^{-2 \pi i (- \frac12 + 2\eta)}$ & 1
  \end{tabular}
  \caption{Eigenvalues of monodromy matrices for the even case}
  \label{tab:eigenMon1}
\end{table}

Since the general structure for both cases is the analogous, we only consider the even case for the discussion.
We start with the case $\eta \in \frac12 \Z$ which follows in the same way as in \cite{W2015IMRN} (following the idea of \cite{O2004} for the
odd case). 

\begin{thm}
  \label{thm:mult2}
  Let $N \in \Z_{\geq 0}$, $\rho \in \{0,1\}$ and
  \[
    \lambda = \frac{2 \sqrt{\alpha \beta (\alpha \beta -1)}}{\alpha+\beta} \left(2 N + \rho +\frac12+2\eta \right)
  \]
  with $\eta \in \frac12 \Z$. If $\lambda \in \Sigma_0$ and $N+1+2\eta>0$, then
  \[
    {\rm dim}_{\C}\{ f \in \mathcal{O}(\Omega) \, | \, H_{\lambda}^{\pm} f = 0 \} = 2.
  \]
\end{thm}

\begin{proof}
  Let us consider the even case. At the singularity $\omega=0$ the difference of coefficients
  is an integer and thus there may be a logarithmic solution. Since the logarithmic solution must correspond to the exponent $\rho=0$, the
  rational solution must have exponent $N$. Now, the sum of exponents of the rational solution must be at least $N + 0 + 0 + (-N)=0$, so the
  solution is unique and must be an scalar multiple of $\omega^{N}$, that is $H_\lambda^{+} \omega^{N}=0$ , but we directly verify that
  \begin{align*}
    H_\lambda^{+} \omega^{N} = \frac{\omega^N\left( -N(N+2\eta)(t-1) - \tfrac12 N  - \bar{q_1} \right)}{\omega(\omega-1)(\omega-t)},
  \end{align*}
  and thus $\bar{q_1}=-N(N+2\eta)(t-1) - \tfrac12 N$, but comparing with the formula of the accessory parameter \eqref{eq:accesory} this
  implies that $\alpha=\beta$ and $\eta=0$ and therefore such a solution does not exist for our current choice of parameters. Therefore, there cannot be a logarithmic solution at $\omega=0$
  and we see that $\mathbf{B}_0 = \mathbf{I}_2$, that is, $\omega=0$ is an apparent singularity.
  
  We proceed similarly for the singularity at $\omega=1$ and we suppose that there is a logarithmic solution. In this
  case there is a meromorphic solution with exponent $N+1+2 \eta>0$ at $\omega=1$ and the sum of exponents must be at least $0 + N+1+2\eta +0 -N = 1+2 \eta$. If the sum of exponents does not vanishes, a rational solution cannot be of such form. In the case that the sum of exponents vanishes, the solution must be a scalar multiple of $(\omega-1)^{N+1+2\eta}$, but
  that cannot be the case by comparing with the accessory parameter as in the case of $\omega=0$. In both cases
  $\mathbf{B}_1=\mathbf{I}_2 = \mathbf{B}_0$ and this shows that the space of holomorphic solutions on $\Omega$ has
  dimension $2$ as desired.
\end{proof}

A simple analysis of the monodromy at the point $\omega=\alpha \beta$ gives the general form of the generators of the solutions at $\Omega$
in both cases.

\begin{cor}
  \label{cor:gensolFin}
  Suppose that $\eta \in \frac12 \Z$ and let $N \in \Z_{\geq0}$ and $\rho \in \{0,1\}$. Further, suppose $\lambda \in \Sigma_0$ is given by
  \[
    \lambda = \frac{2 \sqrt{\alpha \beta (\alpha \beta -1)}}{\alpha+\beta} \left(2 N + \rho +\frac12+2\eta \right)
  \]
  with $N+1+2\eta>0$. Then, there exist Heun polynomials $H p_1(\omega)$
  and $H p_{2}(\omega)$ spanning the space of solutions of $H^{\pm}_\lambda f = 0$ holomorphic at $\Omega$. We have
  \begin{enumerate}
  \item for the even case ($\rho=0$), there are polynomials $f_1(\omega)$ of degree $N$ and $g_2(\omega)$ of degree at
    most $N-1+2\eta$, unique up
    to constant multiples , such that $H p_1(\omega) = f_1(\omega)$ and
    \[
      H p_2(\omega) = \frac{g_2(\omega)\sqrt{\omega-\alpha \beta}}{(\omega- \alpha \beta)^{N+1+2\eta}},
    \]
  \item for the odd case ($\rho=1$), there are polynomials $f_1(\omega)$ of degree $N$ and $g_2(\omega)$ of degree at most $N+2\eta$, unique up
    to constant multiples , such that $H p_1(\omega) = f_1(\omega)$ and
    \[
      H p_2(\omega) = \frac{g_2(\omega)\sqrt{\omega-\alpha \beta}}{(\omega- \alpha \beta)^{N+2+2\eta}}.
    \]
   \end{enumerate}
\end{cor}


  In the case $\eta \in \frac12 \Z$ and $N+1+2\eta\leq 0$, at $\omega=1$ the rational solution corresponds to the exponent $\rho=0$, consistent with the fact that it is a polynomial solution. Therefore the existence, or non-existence, of a logarithmic solution cannot be clarified using the above method.

On the other hand, for the case of $2\eta \notin \Z$, we see that the eigenvalues in $\Sigma_0$ have multiplicity at most $1$ in each parity.

\begin{prop}
  Suppose $\lambda \in \Sigma_0$ with $\eta \notin \frac12 \Z$, then
  \[
    {\rm dim}_{\C}\{ f \in \mathcal{O}(\Omega) \, | \, H_{\lambda}^{\pm} f = 0 \} = 1.
  \]
\end{prop}

\begin{proof}
  We consider only the even case. Let us suppose that the even polynomial solution corresponds to the vector
  $\mathbf{e}_2 = {}^t(0,1)$ in the monodromy representation, then the monodromy matrices $\mathbf{B}_0$ and $\mathbf{B}_1$
  must be of the form
  \[
    \mathbf{B}_0 =
    \begin{bmatrix}
      1 & 0 \\
      c_0 & 1 
    \end{bmatrix},
    \quad
    \mathbf{B}_1 =
    \begin{bmatrix}
      e^{2 \pi i (2\eta)} & 0 \\
      c_1 & 1 
    \end{bmatrix}
  \]
  where $e^{2 \pi i (2\eta)} \neq 1$ for some constants $c_0,c_1 \in \C$. Since the difference of exponents at $\omega=0$ is an integer, we can show as in
  the proof of Theorem \ref{thm:mult2} that $\mathbf{B}_0=\mathbf{I}_2$, making $c_0=0$ and thus $\omega=0$ is an apparent singularity.
  On the other hand, the fact that $e^{2 \pi i (2\eta)} \neq 1$ does not allow the existence of another holomorphic solution at $\omega=1$.
  Indeed, we can take a basis such that the monodromy matrices are given by
  \[
    \mathbf{B}_0 =
    \begin{bmatrix}
      1 & 0 \\
      0 & 1 
    \end{bmatrix},
    \quad
    \mathbf{B}_1 =
    \begin{bmatrix}
      e^{2 \pi i (2\eta)} & 0 \\
      0 & 1 
    \end{bmatrix}
  \]
  with $e^{2 \pi i (2\eta)}\neq 1$, showing that the dimension of the space of holomorphic solutions in $\Omega$ is $1$.
\end{proof}

\begin{rem}
  Note that in the case $\eta\notin \frac12\Z$, by an analogous argument we can show that if there is an eigenvalue
  $\eta \in \Sigma_\infty$ of the form \eqref{eq:finiteEigen} then it must have multiplicity one. Such an eigenvalue would be
  analogous to the non-Juddian exceptional eigenvalues of the AQRM, which always have multiplicity one. In contrast,
  the spectral structure for the case $\eta \in \frac12\Z$ is essentially different to the case of AQRM since both
  finite and infinite eigenvalues are present when there is a degeneracy.
\end{rem}

\begin{rem}
  In \cite{PW2002a}, the analysis of the multiplicity was done for the NCHO using the the action of the Hamiltonian on a twisted basis of $L^2(\R)\otimes \C^2$ 
  (twisted by the matrix $\bK$). The use of the monodromy data allows us to give a shorter and more detailed proof.
\end{rem}

To conclude this section, let us state the converse of  the results of Theorem \ref{thm:mult2}. The proof
follows as in \cite{W2015IMRN} (Theorems 4.3 and 4.4).

\begin{prop}
  Let $\eta \in \frac12 \Z$. Suppose there is a solution of $H^{\pm}_\lambda f = 0$ that is either
  \begin{itemize}
  \item a rational function in $\omega$ at the origin, or
  \item of the form of the form $(\omega-\alpha \beta)^{\frac12} q(\omega)$ at the origin
    for a rational function $q(\omega)$.
  \end{itemize}

  Then, if $\lambda \in \Sigma_0$, we have
  \[
    \dim_\C \{ f \in \mathcal{O}(\Omega) | H_\lambda^{\pm} f = 0 \}  = 2,
  \]
  and the solutions have the form given in Corollary \ref{cor:gensolFin}. \qed
\end{prop}

The final possibility for degeneracy is for an eigenvalue $\lambda$  to be $\lambda\in \Sigma_\infty^{+} \cap \Sigma_\infty^{-}$, that is, a degeneracy
between solutions of infinite type of different parity. In general it is difficult to detect directly these
type of singularities but they are known to exist in the case of the NCHO by the continuity of the spectrum of
the NCHO with respect to the parameters $\alpha,\beta$ and comparison with verified numerical computations
(see \cite{NNW2002}). While not yet verified numerically, it is likely that these types of degeneracies also
occur for the $\eta$-NCHO.

\part{From $\eta$-NCHO to AQRM}
\label{sec:fromsNCHOtoAQRM}

In this part we introduce the special type of confluence procedure that relates the Heun ODE of the $\eta$-NCHO with the confluent Heun picture of the AQRM (see Figure \ref{fig:conf} for the case of the NCHO and the QRM). 

\section{Asymmetric quantum Rabi models}
\label{sec:preliminaries}

The asymmetric quantum Rabi model is the model with Hamiltonian given by
\begin{equation}\label{eq:aH}
  \AHRabi{\eta} = \omega a^\dag a+\Delta \sigma_z +g\sigma_x(a^\dag+a) + \eta \sigma_x,
\end{equation}
where \(\eta \in \R\) is the bias parameter. Here, we use $\eta$ as the bias parameter avoid confusion with the parameter $\epsilon$ defined in the $\eta$-NCHO Heun picture derivation in Section \ref{sec:compl-pict-eigenv}. It is known that when the bias parameter takes value $\eta \in \frac12 \Z$ the spectrum is degenerate and there is a hidden symmetry operator $J_{\epsilon}$ that commutes with $\AHRabi{\eta}$. The (symmetric) QRM corresponds to the case $\eta=0$.

The parity (or $\Z_2$-symmetry) in the case of QRM is well-defined and is based on the existence of the
involution $J_0 = \mathcal{P} \sigma_z$ where $\mathcal{P} = \exp(i \pi a^\dag a)$ commuting with the QRM Hamiltonian. In contrast, when $0 \neq 2\eta \in \Z_{\geq0}$ the commuting
operator $J_{2\eta}$ is not an involution.
In fact, we verify that
\begin{align}
  \label{eq:hiddenS}
  J_{2\eta}^2 = p_{2\eta}(H_{\eta};g,\Delta),
\end{align}
for a non-zero polynomial $p_{2\eta}(x;g,\Delta)$ and where $H_{\eta}$ is the AQRM Hamiltonian. If the kernel of the
operator $J_{2\eta}$ is trivial, or equivalently, if the polynomial $p_{2\eta}(x;g,\Delta)$ does not vanish for
$x = \lambda \in \Spec(H^{(\eta)})$ it is possible to introduce a well-defined parity for this case (i.e. we may define an involution by a suitable normalization of $J_{2\eta}$). While it is conjectured that this is the case, a proof of this fact has not been obtained. For an extended discussion and partial results  we refer the reader to \cite{RBW2021,RBW2022}.

The confluent Heun picture of the AQRM is obtained directly from the Bargmann space $\mathcal{B}$ realization of boson
operators (see e.g. \cite{Sc1967AP}), that is
\[
a \mapsto \partial_z \text{ and } a^\dag \mapsto z
\]
where \( \partial_z = \frac{\partial }{\partial z} \). Denoting the eigenvalue (energy) by $E \in \R$, the eigenvalue
problem for the AQRM is immediately seen to be equivalent to finding entire solutions for confluent Heun differential equation
\begin{equation}
  \label{eq:H1eps}
  \mathcal{H}_1^{\eta}(E) = \frac{d^2}{d y^2} + \left( -4 g^2  + \frac{\alpha+1}{y} + \frac{\alpha- 2\eta}{y-1} \right) \frac{d}{d y} + \frac{- 4g^2 \alpha y +  \mu + 4\eta g^2 - \eta^2 }{y(y-1)}.
\end{equation}
and equivalently, by
\begin{equation}
  \label{eq:H2eps}
  \mathcal{H}_2^{\eta}(E) = \frac{d^2}{d \bar{y}^2} + \left( -4 g^2  + \frac{\alpha - 2 \eta}{\bar{y}} + \frac{\alpha+1}{\bar{y}-1} \right) \frac{d}{d \bar{y}} + \frac{- 4g^2( \alpha - 2\eta + 1) y +  \mu - 4\eta g^2 - \eta^2 }{\bar{y}(\bar{y}-1)},
\end{equation}
with $\alpha = -(E+g^2-\eta)$ and \(\mu\) given by
\[
  \mu = (E + g^2)^2 - 4g^2 (E+g^2) - \Delta^2.
\]

It is important to note that the solutions of the confluent Heun ODE \eqref{eq:H1eps} and \eqref{eq:H2eps} are not directly the eigenfunctions of the AQRM. In fact, the eigenfunctions are given by expressions of the form
\begin{align}
  \label{eq:formEigenAQRM}
  f(z) = e^{\pm g z} \phi(z)
\end{align}
for $z = \frac{(g\pm y)}{2g}$ and where $\phi(z)$ is a solution of \eqref{eq:H1eps} or \eqref{eq:H2eps}. It is also not difficult to verify
from general theory that the solutions $f(z)$ behave asymptotically as $e^{c z}$ for some constant $c \in \C$ and therefore these
are finite with respect to the Bargmann norm and thus it is enough to verify that the solutions are entire. For more details we
refer the reader to \cite{B2013MfI,W2016JPA}.

The existence of two equivalent confluent Heun ODE corresponds to the expansion around the finite singularities
(either $z=-g$ or $z=g$ before normalization) during the derivation of the equations. The existence of the two confluent Heun ODE
related by a $\Z_2$-symmetry was fundamental in the proof of the exact solvability of the QRM by Daniel Braak \cite{B2011PRL}.

Similar to the case of the $\eta$-NCHO, there is a second order element of $\mathcal{U}(\mathfrak{sl}_2)$ that captures
the confluent Heun picture of the AQRM through a particular representation. For \(a \in \C\) define the algebraic
action \(\omega_a\) of $\mathfrak{sl}_2$ on the vector spaces $\bm{V}_{1}:= y^{-\frac14} \C[y, y^{-1}]$ and
$\bm{V}_{2}:=y^{\frac14} \C[y, y^{-1}]$  given by
\begin{gather*}
  \omega_a(H) :=2y\partial_y+\frac12,\quad
  \omega_a(E) :=y^2\partial_y+\frac12(a+\frac12)y,\quad
  \omega_a(F) := -\partial_y+\frac12(a-\frac12)y^{-1}
\end{gather*}
with \(\partial_y := \frac{d}{d y} \). It is not difficult to verify that these operators indeed act on the
space $\bm{V}_j (j=1,2)$, and define infinite dimensional representations of $\mathfrak{sl}_2$.

The representation $\omega_a$ is actually related to the representation $\varpi_a$ introduced in Section \ref{sec:compl-pict-eigenv}. In particular, when considered with respect to the variable
$\omega = z^2 \cth(\kappa) = z^2 \sqrt{\alpha \beta}$ we verify the following relations
\begin{align}
  \label{eq:reprel}
  \varpi_a(H) &= \omega_{a-\frac12}(H), \qquad \varpi_a(E) = \frac{1}{\sqrt{\alpha\beta}} \omega_{a-\frac12}(E), \qquad\varpi_a(F) = \sqrt{\alpha \beta} \omega_{a-\frac12}(F).
\end{align}

Let  $(p, q, r, C) \in \R^4$.  Define a second order element ${\mathbb{K}}=\mathbb{K}(p, q, r; C) \in \mathcal{U}(\mathfrak{sl}_{2})$
and a constant $\lambda_a=\lambda_a(p, q, r)$ depending on the representation $\omega_a$ as follows:
\begin{align*}
  \mathbb{K}(p, q, r; C):= & \left[\frac{1}{2}H-E+p \right]\left(F+q\right)
                                           + r\left[H-\frac{1}{2}\right]+C,\\
  \lambda_a(p, q, r):= & q\left(\frac{1}{2}a +p\right)+r\left(a-\frac{1}{2}\right).
\end{align*}
Noticing $y^{-\frac12(a-\frac12)}\,y\partial_y \,y^{\frac12(a-\frac12)}= y\partial_y+ \frac12(a-\frac12)$, we
obtain
\begin{align*}
  &\frac{y^{-\frac12(a-\frac12)}\omega_a(\mathbb{K}(p, q, r; C))y^{\frac12(a-\frac12)}}{y(y-1)}       \\
  &\qquad \qquad =  \frac{d^2}{dy^2} +\Big\{-q + \frac{\frac12a+p}{y} + \frac{\frac12a+2r-p}{y-1} \Big\}\frac{d}{d y}
    +  \frac{-p q y+\lambda_a(p, q, r)+C}{y(y-1)}. 
\end{align*}

Now, by choosing suitable parameters $(p, q, r; C)$ we define from $\mathbb{K}(p, q, r; C)$ two second order
elements $\mathcal{K}$ and $\tilde{\mathcal{K}}$ $\in \mathcal{U}(\mathfrak{sl}_2)$ that capture the Hamiltonian \(\AHRabi{\eta} \) of the AQRM.

Let $\lambda$ be an eigenvalue of $\AHRabi{\eta}$ and set $\delta=-(E+g^2-\eta)$, $\delta'=\delta-2\eta+1$ and
\[
  \mu = (E + g^2)^2 -4g^2 (E + g^2) - \Delta^2.
\]
Then, by defining
\begin{align*}
  {\mathcal{K}} &:= \mathbb{K}\Big(1+\frac{\delta}2, 4g^2, \frac{\delta'}2\,;\, \mu+4\eta g^2 -\eta^2\Big) \in \mathcal{U}(\mathfrak{sl}_2), \\
  \Lambda_\delta &:= \lambda_\delta\Big(1 + \frac{\delta}2, 4g^2, \frac{\delta'}2\Big),
\end{align*}
we obtain
\begin{equation}
  y(y-1)\mathcal{H}_1^{\eta}(E)= y^{-\frac12(\delta-\frac12)}(\omega_\delta(\mathcal{K})-\Lambda_\delta)y^{\frac12(\delta-\frac12)}.
\end{equation}

Similarly, by defining
\begin{align*}
  \tilde{\mathcal{K}} &:= \mathbb{K}\Big(-1+\frac{\delta'}2, 4g^2, \frac{\delta}2\,;\, \mu-4\eta g^2 -\eta^2\Big) \in \mathcal{U}(\mathfrak{sl}_2), \\
  \tilde{\Lambda}_{\delta'} &:= \lambda_{\delta'}\Big(-1+\frac{\delta'}2, 4g^2, \frac{\delta}2\Big).
\end{align*}
we obtain
\begin{equation}
  \bar{y}(\bar{y}-1)\mathcal{H}_2^{\eta}(E)= \bar{y}^{-\frac12(\delta'-\frac12)}(\omega_{\delta'}(\tilde{\mathcal{K}})-\tilde{\Lambda}_{\delta'})\bar{y}^{\frac12(\delta'-\frac12)}.
\end{equation}

Let us briefly describe the structure of the spectrum of the AQRM and the degeneracy picture. For $N\geq 0$, eigenvalues $E$ of the AQRM of the form
\[
  E = N \pm \eta -g^2,
\]
are called exceptional eigenvalues. Exceptional eigenvalues corresponding to polynomial Frobenius solution of the confluent Heun ODE \eqref{eq:H1eps} or \eqref{eq:H2eps} are called Juddian eigenvalues and the associated solutions are called Juddian, or quasi-exact, solutions. We note here, that the general form of the Juddian eigenfunctions for the AQRM is that of polynomial multiplied by an exponential factor.

The constraint relation for a Juddian eigenvalue \(E = N \pm \eta -g^2\) is given by
\[
  \cp{N,\pm\eta}{N}((2g)^2,\Delta^2) = 0
\]
where $\cp{N,\pm\eta}{N}(x,y)$ is the constraint polynomial of the AQRM.

For $i=0,1,2,\cdots,N$, define the polynomials $\cp{N,\eta}{i}(x,y)$ by the three-term recurrence relation
\begin{align*}
  \cp{N,\eta}{0}(x,y) &= 1, \\
  \cp{N,\eta}{1}(x,y) &= x + y - 1 - 2\eta, \\
  \cp{N,\eta}{k}(x,y) &= (k x + y - k(k + 2\eta) ) P_{k-1}^{(N,\eta)}(x,y) - k(k-1)(N-k+1) x P_{k-2}^{(N,\eta)}(x,y),
\end{align*}
for \(k \geq 2 \). The constraint polynomial is given by the $N$-th polynomial in this family.

For $\eta \notin \frac12 \Z$, there are no degenerate eigenvalues in the spectrum of the AQRM. The degeneracy of Juddian eigenvalues for the QRM was first shown by Ku\'s \cite{K1985JMP}. The presence of degeneracies for the case $\eta = \frac12$ was first observed in \cite{LB2015JPA,LB2016JPA} and proved in \cite{W2016JPA}.

More generally, for $\eta \in \frac12 \Z$ all Juddian eigenvalues are degenerate and, conversely, all degenerate eigenvalues are given
by Juddian eigenvalues. The degeneration of the Juddian eigenvalues is described by the divisibility relation between constraint
polynomials
\begin{align}
  \label{eq:div}
  \cp{N+\ell,-\ell/2}{N+\ell}(x,y) =  A^\ell_N(x,y) \cp{N,\ell/2}{N}(x,y)
\end{align}
for $N,\ell\ge 0$ and where \( A^\ell_N(x,y)\) is a polynomial satisfying \( A^\ell_N(x,y) >0\) for $x, y > 0$. We direct the reader to
\cite{KRW2017} for a complete description of the degeneracy structure of the AQRM and the proof of the above results.
\subsection{Curves determined by the constraint polynomials of the AQRM}
\label{sec:ccpolyAQRM}

For reference, let us give some examples of constraint polynomials of the AQRM for small values of $k$.

\begin{ex}
  For \(k=2,3\), we have
  \begin{align*}
    \cp{N,\eta}{2}(x,y) &= 2 x^2 + 3 x y + y^2 - 2( N +  2(1+2\eta)) x - (5+6\eta) y + 4(1+3\e + 2\eta^2), \\
    \cp{N,\eta}{3}(x,y) &= 6 x^3 + 11 x^2 y + 6 x y^2 + y^3 - 6(2 N + 3(1+2\eta))x^2 - 2(7+6\e) y^2 \\
                      &\quad{} - 2 (4 N +  17 + 22\eta) x y  + 6 (2 N + 3(1+2\eta))(2+2\eta) x \\
                      &\quad{} + (49 + 4\e (24 + 11\eta)) y - 6(1+2\eta)(2+2\eta)(3+2\eta).
  \end{align*}
\end{ex}

In Figures \ref{fig:curve1aqrm} and \ref{fig:curve2aqrm}  we give some examples of the curves determined by the
constraint relation
\[
  \cp{N,\e}{N}((2g)^2,\Delta^2) = 0,
\]
in the $(g,\Delta)$-plane. Note that the curve consist of a number of disconnected curves of elliptic shape.

\begin{figure}[h]
  \centering
  \subfloat[$N=2$]{
    \includegraphics[height=4cm]{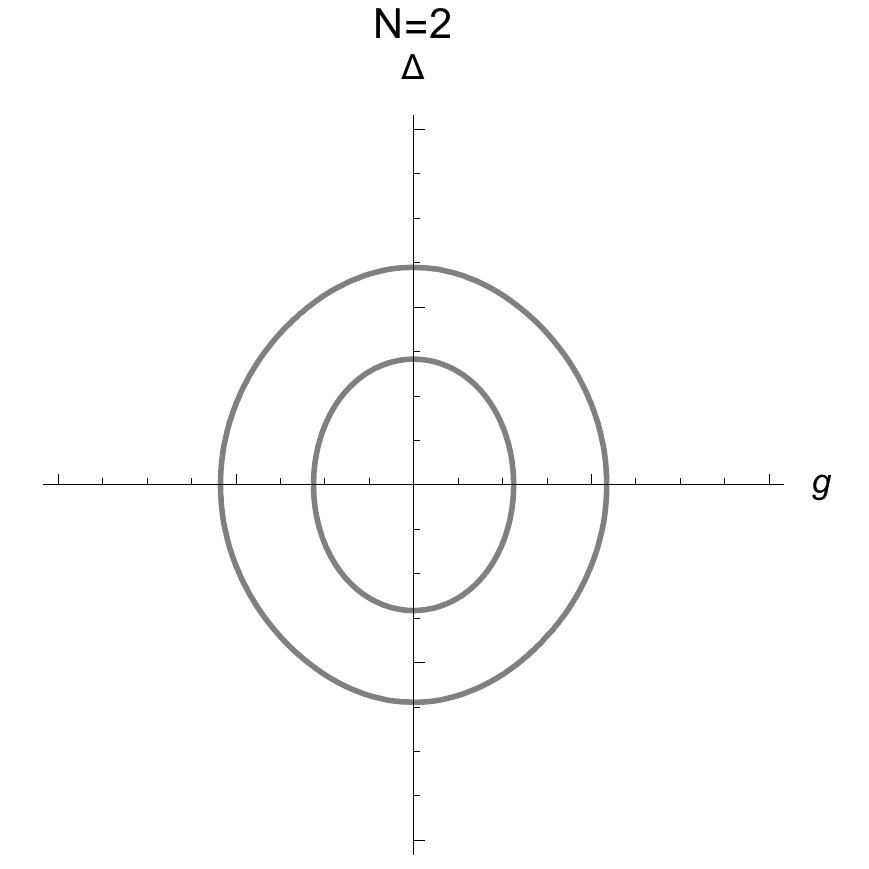}}
  ~ \qquad \qquad
  \subfloat[$N=3$]{
    \includegraphics[height=4cm]{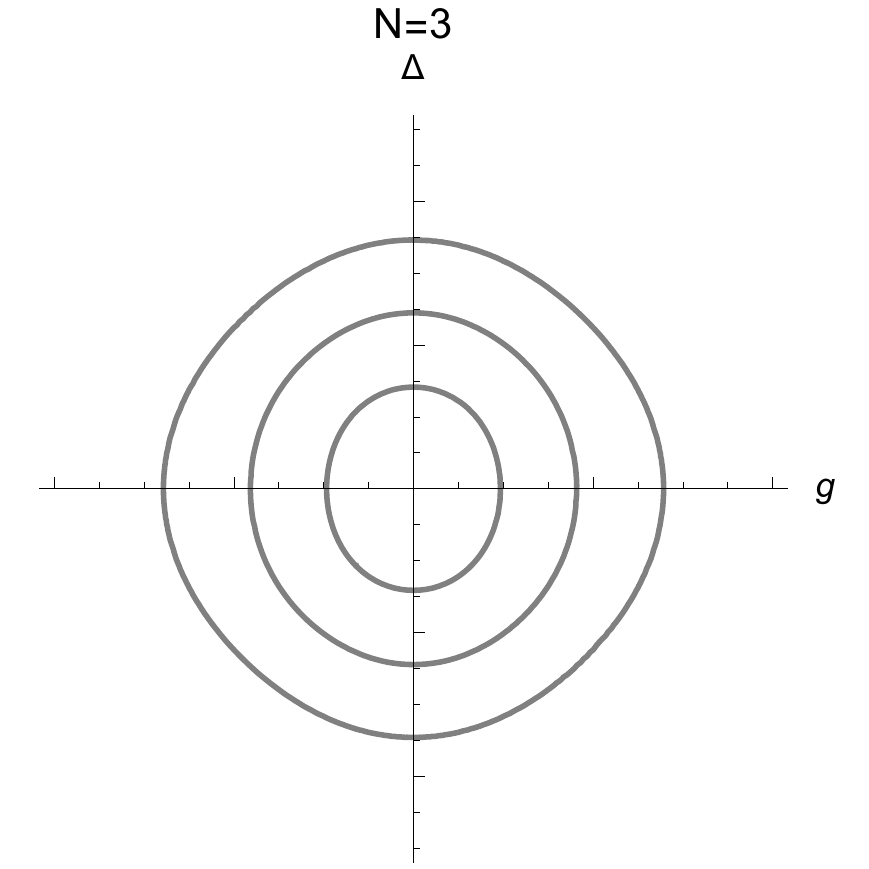}}
  \caption{Curves determined by the constraint condition for $\eta=\frac12$ and $N=2,3$}
  \label{fig:curve1aqrm}
\end{figure}

\begin{figure}[h]
  \centering
  \subfloat[$N=5$]{
    \includegraphics[height=4cm]{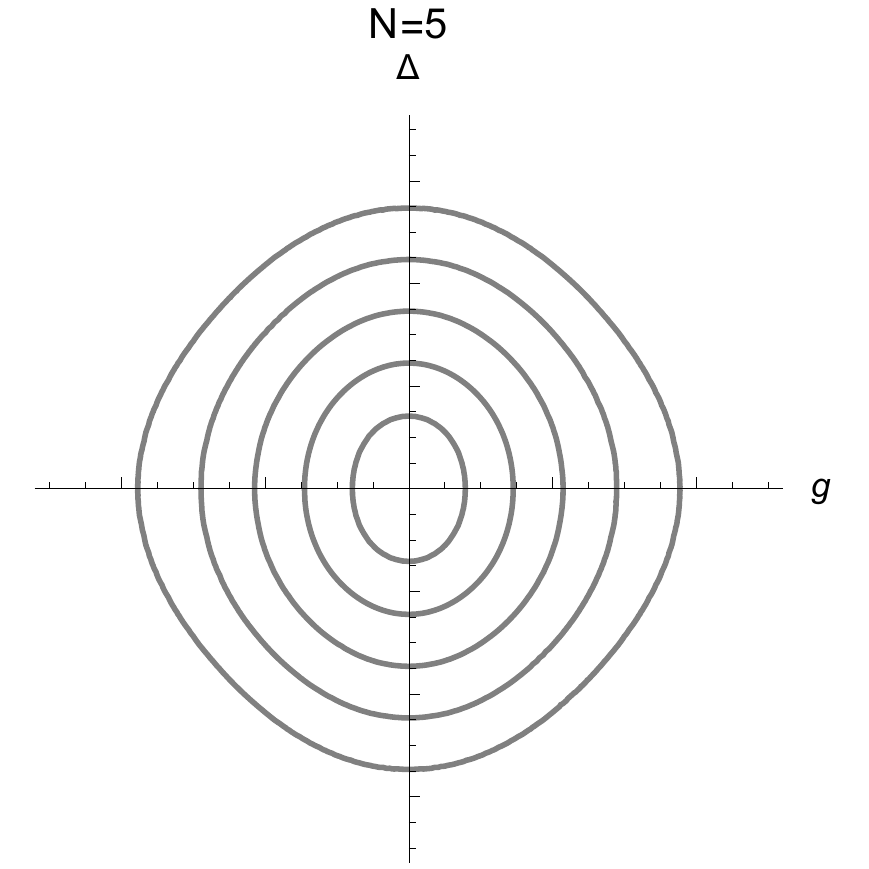}}
  ~ \qquad \qquad
  \subfloat[$N=5$ (detail)]{
    \includegraphics[height=4cm]{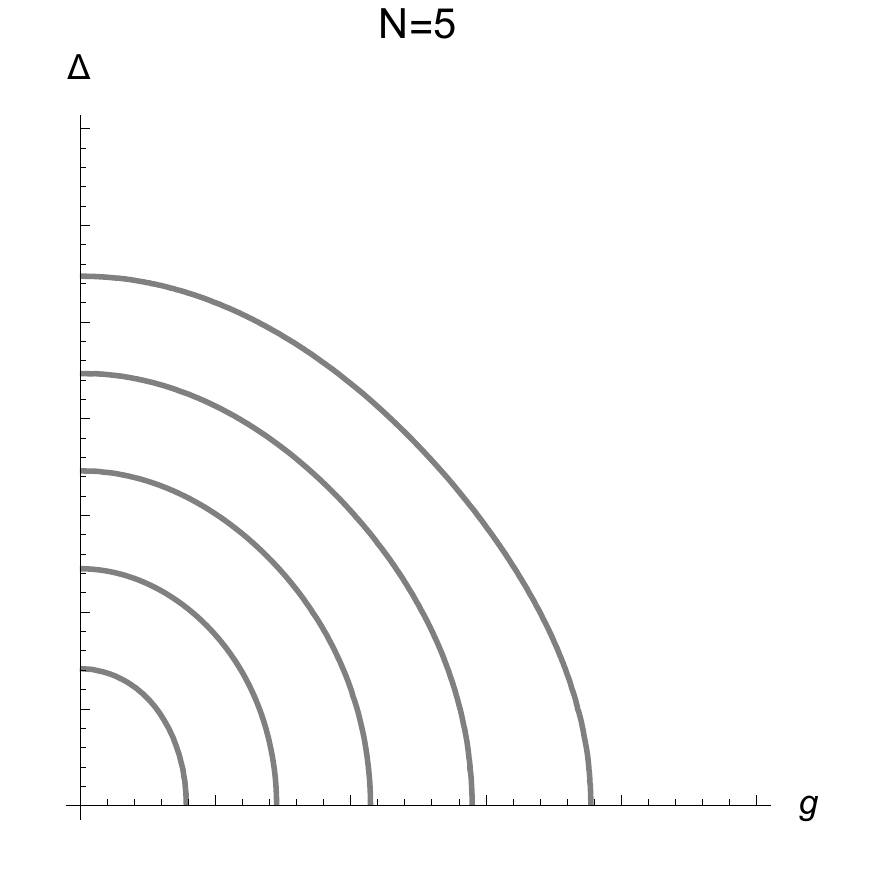}}
  \caption{Curves determined by the constraint condition for $\eta=\frac12$ for $N=5$ and detail}
  \label{fig:curve2aqrm}
\end{figure}


\section{Iso-parallel confluence process: a covering of models}
\label{sec:cprod}

The confluence process that connects the $\eta$-NCHO and the AQRM is based on the standard confluence process that makes the singularity at $\omega=\alpha \beta$
coalesce with the one at $\omega = \infty$ (see e.g. \cite{R1995,SL2000}).

Let us begin by considering the a standard confluence process. Starting from the standard form \eqref{eq:heunStd} of the
Heun ODE, by dividing by $\alpha \beta$ we obtain
\begin{align*}
  &\frac{1}{t} \omega(\omega-1)(\omega-\alpha \beta) \Lambda^a(\omega, \partial_\omega)f(\omega) \\
  & \qquad = \omega(\omega-1)(\tfrac{\omega}{\alpha \beta}-1) f''(\omega) + \left( (\omega-1)(\tfrac{\omega}{\alpha \beta}-1) \bar{C}  + \omega(\tfrac{\omega}{\alpha \beta}-1)\bar{D} + \omega(\omega-1)\tfrac{\bar{F}}{\alpha \beta} \right) f'(\omega)  \\
  & \qquad\qquad \qquad \qquad \qquad \qquad    + \left(\bar{A}\tfrac{\bar{B}}{\alpha \beta}\omega - \tfrac{\bar{q}_a}{\alpha \beta}\right) f(\omega) = 0,
\end{align*}
and the limit $\alpha \beta \to \infty$ gives a confluent Heun ODE as long as the limits
\[
  \lim_{ \alpha \beta \to \infty} \frac{\bar{B}}{\alpha \beta}, \qquad  \lim_{ \alpha \beta \to \infty} \frac{\bar{F}}{\alpha\beta}, \qquad \lim_{t\to \infty} \frac{\bar{q}_a}{\alpha \beta},
\]
are finite. In addition, to obtain the particular confluent Heun ODE of the AQRM (see Section \ref{sec:preliminaries}) we require
\[
  \lim_{\alpha \beta\to \infty} \frac{\bar{B}}{\alpha \beta}=  \lim_{\alpha \beta\to \infty} \frac{\bar{F}}{\alpha \beta}.
\]

In \cite{W2015IMRN} this was achieved by introducing a particular change of variable on the parameters $a$ and $\nu$ before
making the confluence procedure. Following that idea we introduce the following definition.

\begin{dfn}
  \label{dfn:ipc}
  The {\em iso-parallel confluence process} of the $\eta$-NCHO is the confluence process obtained by introducing the change
  of variables
  \begin{equation}
    \label{eq:cvp}
    a \mapsto a + p, \qquad \nu \mapsto \nu + p
  \end{equation}
  in the Heun ODE determined by \eqref{eq:HeunA1}, or \eqref{eq:HeunA2}, and taking $\alpha \beta \mapsto \infty$. In addition, we require that
  \[
    p = \alpha \beta r + o\left(\alpha \beta\right)
  \]
  for some $r\in \C$ and that there exists $k\in \R$ such that 
    \[
    \left|\frac{\alpha-\beta}{\alpha+\beta}\right| = k(\alpha \beta)^{-1} + o\left((\alpha \beta\right)^{-2}) .
  \]
The second requirement in particular implies that $\frac{\alpha}{\beta}\to 1$ when $\alpha \beta\to\infty$

  Note that when considering $\lambda \in \Sigma_0$ of the form
  \[
    \lambda = \frac{2 \sqrt{\alpha \beta (\alpha \beta -1)}}{\alpha+\beta} \left(L + \tfrac12 + 2 \eta\right),
  \]
  that is, with $\nu = L + \frac12 + 2\eta$, the transformation \eqref{eq:cvp} is replaced by
  \[
    a \mapsto a + p, \qquad L \mapsto L + p.
  \]
\end{dfn}

The change of variable in the iso-parallel confluence makes the parameters $a$ and $\nu$ go to infinity along with
$\alpha \beta$ in a parallel way, explaining the name given to the process. The parameters $k$ and $r$ control the asymptotic growth
of the parameters $\alpha,\beta$ under the confluence process.

Next, we determine the values to obtain the appropriate confluent Heun ODE corresponding to the AQRM.
The change of variable of the iso-parallel confluence process gives
\begin{gather*}
  \bar{B} \mapsto \bar{B} + p, \quad \bar{F} \mapsto \bar{F} + p,
\end{gather*}
with no changes to $\bar{A}$, $\bar{C}$ or $\bar{D}$. According to the asymptotics
we must have
\[
  \lim_{ \alpha \beta \to \infty} (\bar{B} + p )(\alpha \beta)^{-1} =  \lim_{\alpha \beta \to \infty} (\bar{F} + p)(\alpha \beta)^{-1} = r.
\]

Now, let us consider the condition on the accessory parameter $\bar{q}_a$, that is
\[
  \lim_{\alpha \beta \to \infty}  
     \left[ - (a-\frac12 - \nu)^2 + (2\eta)^2 + (\epsilon (\nu+p))^2 \right] (1 - \tfrac{1}{\alpha \beta}) - \tfrac{2}{\alpha \beta} (a + p -\tfrac12)(a - \tfrac12 - \nu+2\eta),
\]
using the asymptotic condition for $\epsilon=\left|\frac{\alpha-\beta}{\alpha+\beta}\right|$ given by $k$ we obtain
\[
  - (2A)^2 + (2\eta)^2 + k^2r^2 - 4 r A - 4 r \eta,
\]
so that the limit is finite. The resulting ODE a confluent Heun differential equation given by
\begin{align}
  \label{eq:conf1}
  \Bigg[\frac{d^2}{d \omega^2} &+ \left( - r + \frac{1+A+\eta}{\omega} + \frac{A-\eta}{\omega-1}\right) \frac{d}{d \omega} \nonumber \\
  &\qquad + \frac{ -r(A+\eta) \omega - (2 A)^2 + (2\eta)^2- 4 r A + (r k)^2 - 4 r \eta}{\omega(\omega-1)} \Bigg] \varphi = 0.
\end{align}

The same considerations, for the Heun ODE given by $\bar{\Lambda}^a(\omega, \partial_\omega) \phi = 0$ results in confluent Heun ODE 
\begin{align}
  \label{eq:conf2}
  \Bigg[\frac{d^2}{d \omega^2} &+ \left( - r + \frac{A+\eta}{\omega} + \frac{1+A-\eta}{\omega-1}\right) \frac{d}{d \omega} \nonumber \\
  &\qquad + \frac{-r(A+1+\eta) \omega - (2 A)^2 + (2\eta)^2- 4 r A + (r k)^2 - 4 r \eta}{\omega(\omega-1)} \Bigg] \varphi = 0.
\end{align}

Summarizing the discussion, we have the following result.

\begin{prop}
  \label{prop:confl}
  For the Heun ODE determined by \eqref{eq:HeunA1} corresponding to the eigenvalue problem of the $\eta$-NCHO,
  the confluence of the singular points $\omega=\alpha \beta$ and $\omega = \infty$ given by the iso-parallel confluence process
  (given in Definition \ref{dfn:ipc}) results in the confluent Heun ODE \eqref{eq:conf1}. Then, the changes of variable
  \[
    r = 4 g^2 \qquad A = - (E + g^2),
  \]
  and
  \[
    k^2= \frac{5\{(E+g^2)^2-4g^2(E+g^2)+4g^2(\eta-\eta^2)\}-\Delta^2}{(4g^2)^2}
  \]
  gives the ODE picture \eqref{eq:H1eps} of the AQRM.

  Similarly, the iso-parallel confluence process to the Heun ODE equation determined by \eqref{eq:HeunA2} corresponding to the twisted
  $\eta$-NCHO gives the confluent Heun ODE equation \eqref{eq:H2eps} of the AQRM. The change of variable is the same as the
  case above with the addition of
  \[
    \eta \mapsto - \eta.
  \]  \qed
\end{prop}

To finish this section we discuss the effect of the iso-parallel confluence process on eigenvalues of $\Sigma_0$ and the degenerate solutions
of the AQRM.

First, let us clearly state the relation between the parameters of the Heun ODE of the $\eta$-NCHO and those of the confluent ODE
of the AQRM. Recall that $A= \frac14\left( -1-2\nu + 2 a \right)$. In the confluence process for equation \eqref{eq:HeunA1} we have the
change of variable
\[
  \bar{A} = A + \eta \mapsto -(E + g^2-\eta) 
\]
and in the confluence process for \eqref{eq:HeunA2} we use the change of variable
\[
  \bar{A} = A + \eta \mapsto - (E + g^2) -\eta  = -(E+g^2-\eta) - 2\eta 
\]
where in both cases $\alpha = -(E+g^2-\eta)$ is the parameter appearing in \eqref{eq:H1eps} and \eqref{eq:H2eps}. Note that in
particular, the change of variable is natural with respect to the ``shifted'' parameter  $\bar{A}$.

Moreover, suppose that we have an eigenvalue $\lambda \in \Sigma_0$ of the form
\[
  \lambda = \frac{2 \sqrt{\alpha \beta (\alpha \beta -1)}}{\alpha+\beta} \left(L + \tfrac12 + 2 \eta\right)
\]
for some $L \in \Z_{\ge0}$ and with $a$ chosen according to parity.
Then, under the iso-parallel confluent process the eigenvalue associated to the Heun ODE \eqref{eq:HeunA1}
corresponds in the AQRM confluent Heun ODE picture to
\begin{align}
  \label{eq:descentEigen}
  E = \frac12 (L+1-a) -g^2 + \eta.
\end{align}
Note in particular that the eigenvalues in $\Sigma_0$ associated to the parameters $L=2 M$ and $L= 2M + 1$ for $M \ge 1$ under the confluence
process correspond to the same Juddian eigenvalue
\[
  E = M+1 -g^2 +\eta.
\]
Similarly, for the \eqref{eq:HeunA2} the corresponding eigenvalue for the AQRM confluent Heun ODE picture
is given by
\[
  E = \frac12 (L+1-a) -g^2 - \eta.
\]
Note that in both formulas for $E$ above, the parameter $a$ corresponds to the initial value of $a$ before the confluence process and
it is determined by the parity of $L$.

\begin{prop}
  \label{prop:confl2}
  Let $\eta \in \frac12 \Z_{\ge 0}$, then for $L \in \Z_{\ge 0}$ the eigenvalues $\lambda \in \Sigma_0$ of $Q^{(\eta)}$ of the form
  \[
    \lambda = \frac{2 \sqrt{\alpha \beta (\alpha \beta -1)}}{\alpha+\beta} \left(L + \tfrac12 + 2 \eta\right)
  \]
  and the eigenvalue $\bar{\lambda} \in \Sigma_0$ of  $\bm{K} Q^{(\eta)} \bm{K}$ of the form
  \[
    \bar{\lambda} = \frac{2 \sqrt{\alpha \beta (\alpha \beta -1)}}{\alpha+\beta} \left(L + 4\eta + \tfrac12 + 2 \eta\right)
  \]
  corresponds under the confluence process given in Proposition \ref{prop:confl} to the degenerate Juddian eigenvalue
  \[
    E = \frac12 (L+1-a) -g^2 + \eta, 
  \]
  where $a=1$ for $L$ even and $a=2$ for $L$ odd. \qed
\end{prop}

Thus, it may appear that up to four potential eigenfunctions corresponding to eigenvalues of $\Sigma_0$ coalesce onto two Juddian
solutions under the confluent process. In fact, due to the possibility of degeneracy we may have up to eight eigenfunctions
(see Section \ref{sec:multiplicity}).

As a first step to clarify this situation, we note that if there is a non-polynomial holomorphic solution $H p(\omega)$ (given by certain
Hermite polynomial) corresponding to an eigenvalue $\lambda \in \Sigma_0$, a direct computation from the three-term recurrence relation of
$H p(\omega)$ shows that the solution vanishes under the iso-parallel confluence process, so that $H p(\omega)$ converges to the trivial
constant $0$ solutions.

To further clarify the relation between eigenvalues $\lambda \in \Sigma_0$ of the $\eta$-NCHO and the Juddian eigenvalues of AQRM, in Part \ref{sec:extendeddiscss} we study the constraint polynomials under the confluence process 

\begin{rem}
  Let us return to the relation given in \eqref{eq:reprel} between the representations $\varpi_a$ and $\omega_a$ in light of the confluence
  process described in this section. Setting $\delta = -(E+g^2 -\eta)$ and with $a$ determined by \eqref{eq:descentEigen}, we have up to
  constants, the equivalence between the representations
  \[
    \varpi_{\frac{a-L}{2}} = \varpi_{\delta + \frac12}  \leftrightarrow \omega_\delta.
  \]
\end{rem}

\part{Extended discussion on the confluence picture}
\label{sec:extendeddiscss}

In this part we investigate with more detail the relation between the polynomial eigenfunctions corresponding to $\lambda \in \Sigma_0 $ and the Juddian solutions of the AQRM  that occurs for certain extraordinary conditions on the parameters.

Under the effect of the iso-parallel confluence process, the quasi-exact solutions of the $\eta$-NCHO descend exactly into Juddian solutions
of the AQRM when the parameters satisfy certain additional constraint.
We also consider the opposite situation, where given parameters corresponding to a Juddian solution of the AQRM, we describe the exceptional conditions given by an additional constraint \eqref{eq:const3} that make the Juddian solution descend from a quasi-exact solution of the $\eta$-NCHO and the effect on the corresponding constraint relations.

While we do not consider the situation for the whole spectrum, even the quasi-exact solutions alone may give considerably information of the complete structure of the spectrum as in the case of the AQRM. We summarise the situation of the AQRM in Section \ref{sec:cpolysig}. Finally, in Section \ref{sec:comparison} as a conclusion and summary, we compare the $\eta$-NCHO and the AQRM and their spectra.

\section{Descent of quasi-exact solutions of $\eta$-NCHO}
\label{sec:cpoly}

We note from Proposition \ref{prop:confl} that in order to obtain a valid confluent Heun ODE picture of an AQRM there are five parameters involved. These parameters are the shift parameter $\eta$, the iso-parallel confluence parameters $(r,k)$ and the AQRM parameters $(g,\Delta)$.

In the following discussion, we fix the parameter $\eta$ and $L \in \Z_{\geq0}$ and consider an eigenvalue $\lambda \in \Sigma_0$ of the form
\[
  \lambda = \frac{2 \sqrt{\alpha \beta (\alpha \beta -1)}}{\alpha+\beta} \left(L + \tfrac12 + 2 \eta\right)
\]
For simplicity, in this section we set $N$ such that $L = 2 N$ or $L = 2 N +1$.

Let us start with the case where the parameter $g$ is fixed. From Proposition \ref{prop:confl} we see that $r =  (2g)^2$
and thus $r$ is also fixed. Then, observing that the change of variable $A = - (E + g^2)$ the value of $k$ is given by
\begin{align}
  \label{eq:constraK}
  k^2= \frac{5\{(E+g^2)^2-4g^2(E+g^2)+4g^2(\eta-\eta^2)\}-\Delta^2}{(4g^2)^2},
\end{align}
and thus for arbitrary $\Delta>0$, $k$ is completely determined (up to sign) by the eigenvalue $E$ of the AQRM.

The proof of the next result follows by direct computation (see also the proof of Theorem \ref{thm:constConflucence2} below).

\begin{prop}
  \label{prop:constConflucence1}
  With the notation of \ref{thm:finite}. Suppose that $\eta$ and $g$ are fixed and let $p(\omega;\alpha,\beta)$ be a quasi-exact solution
  corresponding to and eigenvalue $\lambda \in \Sigma_0$ of the form
  \[
    \lambda = \frac{2 \sqrt{\alpha \beta (\alpha \beta -1)}}{\alpha+\beta} \left(L + \tfrac12 + 2 \eta\right).
  \]
  Then, under the iso-parallel confluence procedure, if the parameters $r$ and $\Delta$ satisfy the condition
  \[
    k^2= \frac{5\{(E+g^2)^2-4g^2(E+g^2)+4g^2(\eta-\eta^2)\}-\Delta^2}{(4g^2)^2},
  \]
  the eigenfunction $p(\omega;\alpha,\beta)$ converges to a non-zero Juddian solution of the AQRM denoted by $p(\infty:(2g)^2,\Delta^2)$
  corresponding to the eigenvalue $E = N -g^2 + \eta$. \qed
\end{prop}

\section{Ascent of Juddian solutions of AQRM}
\label{sec:cpoly2} 

Now, lets consider the complementary situation where we start from a fixed $g$, $\Delta$ and
$\eta$, that is, an specific instance of the AQRM and consider the conditions for the
existence of parameters $r,k$ such that an eigenvalue $\lambda \in \Sigma_0$ of $\eta$-NCHO descends into
the Juddian eigenvalue $E = N -g^2 + \eta$.  We may refer tentatively to this situation by saying that the eigenvalue
$E$ ascends into the eigenvalue $\lambda$, or that it descends from the eigenvalue $\lambda$.

Let us begin by considering the constraint condition \eqref{eq:constcond} for an eigenvalue $\lambda \in \Sigma_0$ of the form \eqref{eq:finiteEigen} as in Theorem \ref{thm:finite}. Let us write the coefficients of the matrix $\bM_L^{(a,\rho)}(\alpha,\beta,\eta)$ in terms of the variables $\alpha, \beta$, concretely,
\begin{align*}
  c_i^{(L)}  =&  \frac1{\sqrt{\alpha \beta}} \Bigg[ \left(\frac{\alpha-\beta}{\alpha+\beta}\right)^2\frac{ (2 L + 1 + 4 \eta)^2(\alpha \beta -1)}{4} \\
  & \quad \quad \quad   + (i - L) \left\{(L-i + 4\eta)\alpha \beta -  (L+1 + i + 4\eta)\right\} \Bigg],
\end{align*}
and we note that $d^{(L)}_{i}$ and $f^{(L)}_{i}$ do not depend on $\alpha$ and $\beta$.

Next, we write the constraint condition \eqref{eq:constcond} as a three-term recurrence relation by the expansion of the continuant.

Namely, setting  $\rho \in \{0,1\}$ according to \( L \equiv \rho \pmod{2} \), define the continuant  $R^{(L,\eta,a,\rho)}_k(\alpha,\beta) $
\begin{equation*}
  R^{(L,\eta,a)}_k(\alpha,\beta) := \det\Tridiag{c^{(L)}_{L-2i}}{d^{(L)}_{L-2i-1}}{f^{(L)}_{L-2i-1}}{1\leq i \leq k},
\end{equation*}
for $k=1,2,\ldots, \frac{L-\rho}{2}$. Then, the constraint condition is given by 
\[
  R_{\frac{L-\rho}2}^{(L,\eta,a)} (\alpha, \beta) = 0
\]
and the three-term recurrence relation
\begin{align*}
  R^{(L,\eta,a)}_{k}(\alpha,\beta) = c_{L-2k}^{(L)} R_{k-1}^{(L,\eta,a,\rho)}(\alpha,\beta) - d^{(L)}_{L- 2k} f^{(L)}_{L-2k}  R_{k-2}^{(L,\eta,a,\rho)}(\alpha,\beta).
\end{align*}
holds. Next, to eliminate the square root we normalize the continuant by setting
\begin{equation}
  \label{eq:normCP}
  q^{(L,\eta,a)}_k(\alpha,\beta) :=  (\tfrac14)^k (\alpha\beta)^{-\frac{k}2}R^{(L,\eta,a)}_k(\alpha,\beta).
\end{equation}
Explicitly, we have  $q^{(L,\eta,a,\rho)}_i(\alpha,\beta) = 1$ and 
\begin{align}
  \label{eq:constPolysNCHO}
  q^{(L,\eta,a)}_i(\alpha,\beta)  &=  \tfrac14 \tilde{c}_{L-2i}^{(L)} q_{i-1}^{(L,\eta,a)}(\alpha,\beta) \nonumber \\
  &\, - \tfrac14(\alpha\beta)^{-1}(L-2i+a)(L-2i-a+3)i(i-1) q_{i-2}^{(L,\eta,a)}(\alpha,\beta), 
\end{align}
for $i=1,2,\ldots, \frac{L-\rho}{2}$ with
\begin{align*}
  \tilde{c}_{L-2i}^{(L)} 
  = \frac{\alpha\beta-1}{4 \alpha\beta} \left(\frac{\alpha-\beta}{\alpha+\beta}\right)^2
  (2 L + 1 + 4 \eta)^2  - 2 i\left\{2(i+ 2\eta) -  \frac{(2L-2 i + 1 + 4\eta)}{\alpha \beta}\right\}.
\end{align*}
Note that by using an appropriate normalization, it is also possible to give a polynomial formulation for the constraint relation (instead of the rations functions $q^{(L,\eta,a)}_k(\alpha,\beta)$) but we do not use it in this paper.

Recall that in this case, the change of variable in the the iso-parallel confluence process is given by
\[
  a \mapsto a + p, \qquad L \mapsto L + p.
\]

\begin{thm}
  \label{thm:constConflucence2}
  Suppose $\lambda \in \Sigma_0$ with 
  \[
    \lambda = \frac{2 \sqrt{\alpha \beta (\alpha \beta -1)}}{\alpha+\beta} \left(L + \tfrac12 + 2 \eta\right),
  \]
  for some integer $L \in \Z_{\ge0}$. The corresponding rational functions $q^{(L,\eta,a)}_i(\alpha,\beta)$ defined in \eqref{eq:constPolysNCHO} gives
  the constraint condition of $\lambda$ by 
  \[
    q_{N}^{(L,\eta,a)} (\alpha, \beta) = 0.
  \]

  Under the iso-parallel confluence process (with $L \mapsto L+p$) with the change of variable
  \begin{equation}
    \label{eq:condCfl}
    r = (2g)^2, \qquad k= \frac{\Delta}{2g^2}
  \end{equation}
  we obtain
  \[
    \lim_{\alpha \beta \to \infty} q^{(L,\eta,a)}_i(\alpha,\beta) = \cp{N,\eta}{i}((2g)^2,\Delta^2)
  \]
  with $N = \frac12(L+1-a) \in \Z_{\geq 0}$. 

  In particular, suppose that the parameters $g,\Delta$ satisfy the constraint condition for the eigenvalue $E = N -g^2+\eta$,
  that is,
  \[
    \cp{N,\eta}{N}((2g)^2,\Delta^2)= 0
  \]
  then if the parameters $r,k$ are given as in \eqref{eq:condCfl} and they satisfy the compatibility condition
  \begin{align}
    \label{eq:const3}
    N(4 g^2) + \Delta^2 - N(N+2\eta)=0,
  \end{align}
  then the eigenvalue $E$ comes from an eigenvalue $\lambda \in \Sigma_0$ for some values of $\alpha$ and $\beta$.
\end{thm}

An illustration of the relation between the constraint relations between the $\eta$-NCHO and the AQRM given in Theorem \ref{eq:condCfl} is given in Appendix \ref{sec:curvesNCHOAQRM}.

\begin{proof}
  Denote by $p^{(L,\eta,a)}_i(r,k)$ the result of applying the iso-parallel confluence process to $q^{(L,\eta,a)}_i(\alpha,\beta)$, that is,
  \[
    p^{(L,\eta,a)}_i(r,k) := \lim_{\alpha\beta\to \infty} q^{(L,\eta,a)}_i(\alpha,\beta),
  \]
  then we have the recurrence relation:
  \begin{align*} 
    p^{(L,\eta,a)}_i(r,k)  &=  \frac14 \{k^2r^2 -4i(i+2\eta-r) \}p_{i-1}^{(L,\eta,a)}(r,k) \\
                        & \quad -  \frac12 r(L-2i-a+3)i(i-1) p_{i-2}^{(L,\eta,a)}(r,k)\\
                        &=  \{i r+\frac{k^2r^2}4- i (i+2\eta) \}p_{i-1}^{(L,\eta,a)}(r,k) \\
                        &\quad- \frac12 i(i-1)(L-2i-a +3) r p_{i-2}^{(L,\eta,a)}(r,k).
  \end{align*}

  Note that, by taking the change of variable $r = (2g)^2$ and taking $ \frac{k^2 r^2}4 = \Delta^2$, that is, 
  \begin{equation}
    \label{eq:const2}
    k = \frac{\Delta}{2g^2},
  \end{equation}
  we obtain the desired result by verifying that the resulting expression is identical with constraint polynomial corresponding to
  the Juddian eigenvalue $E = \frac12(L+a+1)-g^2+\eta$ given by
  \begin{align*}
    \cp{L,\eta}{i}((2g)^2,\Delta^2) &= (i (2g)^2 + \Delta^2 - i(i + 2 \eta) ) P_{i-1}^{(L,\eta)}((2g)^2,\Delta^2) \\
                            &\quad - \frac{1}2 i (i-1)(L-2i-a+3) (2g)^2 P_{i-2}^{(L,\eta)}((2g)^2,\Delta^2).
  \end{align*}
  The compatibility condition \eqref{eq:const3} is required since the parameters must satisfy  simultaneously the changes of variables of
   Proposition \ref{prop:confl} and \eqref{eq:condCfl} needed for the identification of the constraint polynomials.
\end{proof}

\begin{rem}
  With the notation of Theorem \ref{thm:constConflucence2} and proceeding in a similar way
  we verify that  under the iso-parallel confluence process we have
  \[
    \det\left(\widetilde{\bM}_{L+4\eta}^{(a,\rho)}(\alpha, \beta, -\eta)\right)  \to  P_{N+2\eta}^{(N+2\eta, -\eta)}(4g^2, \Delta^2)
  \]
  that together with \eqref{eq:const2} gives the constraint
  \begin{equation}
    \label{eq:const3a}
    (N + 2\eta)(2 g)^2 + \Delta^2 - N(N+2\eta)=0,
  \end{equation}

  Therefore, even though we have the divisibility relation \eqref{eq:div} for constraint polynomials, we see that \eqref{eq:const3}
  and \eqref{eq:const3a} are not compatible unless $g=0$. Thus for $2\eta\in \Z$, the constraints relations \eqref{eq:const3}
  and \eqref{eq:const3a} are consistent only for $g=0$. The reason for this is because of the definition of $k$ as
  $k=\Delta/2g^2$. Taking the definition of $k$ by the formula \eqref{eq:constraK} which depends on the eigenvalues, then
  we have the realization of the degenerate eigenvalues of the AQRM from $\eta$-NCHO. 
\end{rem}

\begin{rem}
  \label{rem:finiteSolmL}
  Let us assume that there is a quasi-exact solution of $\eta$-NCHO. By Proposition 5.1., if $r=4g^2$ and $A=-(N+\eta)$, then as
  far as the constraint relation
  \[
    (4g^2k)^2= 5N^2+10\eta N-20g^2N-\Delta^2
  \]
  holds for the variables $k$ and $\Delta$, there is a (non-zero) Juddian solution of the AQRM (for $(g, \Delta, \eta)$) corresponding the
  eigenvalue $E=N-g^2+\eta$ and is obtained from the polynomial solution of the $\eta$-NCHO. This can be verified directly from observing
  the recurrence relation (at $w=0$) of the polynomial solution of the Heun ODE for $\eta$-NCHO.

  Moreover, we note that
  \begin{enumerate}
  \item  A priori, the recurrence relation of the constraint polynomial of AQRM defined for a given
    $(g, \Delta, \eta)$ does not necessarily takes the form given in Theorem \ref{thm:constConflucence2}.
  \item Both Proposition \ref{prop:confl} and Theorem \ref{thm:constConflucence2} hold for any representation parameter
    $a$ i.e. for any $\eta$-NCHO of type $a$ (determined by $\varpi_a(S) \varphi=0$).
  \end{enumerate}
\end{rem}

We are now in the position to clarify the situation regarding the eigenvalues $\lambda \in \Sigma_0$ and the associated
eigenfunctions under the confluence process.

\begin{thm}
  Let $\eta \in \frac12 \Z_{\ge0}$ with $\eta \neq 0$. Let us suppose that $\lambda \in \Sigma_0$ of the form
  \[
    \lambda = \frac{2 \sqrt{\alpha \beta (\alpha \beta -1)}}{\alpha+\beta} \left(L + \tfrac12 + 2 \eta\right),
  \]
  with $L = 2 N + \rho$ (with the notation of Theorem \ref{thm:finite}) that is, it satisfies the constraint relation
  \[
    \det(\bM_L^{(a,\rho)}(\alpha,\beta,\eta)) = 0.
  \]
  Then, under the iso-parallel confluence $\lambda$ coalesces to the eigenvalue
  \[
    E = N - g^2 + \eta
  \]
  and it is the only eigenvalue $\lambda \in \Sigma_0$ that has this property.
\end{thm}

\begin{proof}
  According to Proposition \ref{prop:confl2} it is enough to verify that eigenvalue $\lambda_1 \in \Sigma_0$ of the form
  \[
    \lambda_1 = \frac{2 \sqrt{\alpha \beta (\alpha \beta -1)}}{\alpha+\beta} \left(L+ 4\eta + \tfrac12 + 2 \eta\right),
  \]
  does not correspond to the eigenvalue $E$ under the confluence process. However, it is clear that
  the parameter $L+ 4 \eta$ does not satisfy the constraint \eqref{eq:const3} if $L$ satisfies it unless $\eta = 0$.
\end{proof}

Note that the additional constraint relation \eqref{eq:const3} appears only via the study of the confluence process between the $\eta$-NCHO
and the AQRM. It is expected that this may provide further insight into the two models, for instance, note that if the parameters $(g,\Delta,\eta)$ are fixed, then $L$ may at most take $2$ values, since \eqref{eq:const3} is a quadratic equation on $L$. This may provide a way for counting the number of Juddian eigenvalues for the AQRM. We leave such considerations for another occasion.

\section{The significance of the constraint polynomials}
\label{sec:cpolysig}

In the previous sections, by means of the iso-parallel confluence process we have shown a
direct relation between the eigenvalues of finite type $\Sigma_0$ of the $\eta$-NCHO and the Juddian eigenvalues, and more generally the exceptional eigenvalues, of the AQRM.

In the current study we do not cover the case of other type of eigenvalues for the two models and their relation, but it may be expected that a similar relation holds between the two since for the AQRM it may be argued that the Juddian solutions determine the full spectrum to a large extent at least in the half-integer bias parameter case. In this section we recall the two most noteworthy examples of this phenomenon, the generalized adiabatic approximation \cite{LB2021b} and the (conjectural) excellent approximation \cite{RBW2022}.

\begin{figure}[h]
  \centering
  \begin{tikzpicture}[
    roundnode/.style={circle, draw=black!60, fill=black!5, very thick, minimum size=7mm},
    squarednode/.style={rectangle, draw=black!60, fill=black!5, very thick, minimum size=5mm},
    ]
\node[squarednode]      (CEE)           [align=center]                   {Conjectural \\ Excelent Approximation \cite{RBW2022}};
\node[squarednode]        (GAA)       [above=of CEE,align=center] {Generalized \\ Adiabatic Approximation \cite{LB2021b}};
\node[squarednode]      (CPoly)       [above right= 5mm and 10mm of GAA,align=center] {Constraint Polynomials \\ $\left\{\cp{N,\frac{\ell}{2}}{N}((2g)^2,\Delta^2)\right\}$ };
\node[squarednode]        (SS)       [below right= 5mm and 13mm of CEE,align=center] {Spectral structure of\\ AQRM ($\eta=\frac{\ell}2$) \cite{KRW2017} };

\draw[-,thick] (GAA.east) -- (3,2.1);
\draw[-,thick] (CEE.east) -- (3,0);
\draw[-,thick] (3,0) -- (3,2.1);
\draw[->,thick] (3,1) -- (5,1) node [midway, below, align=left] (TextNode) {};
\draw[->,shorten >=5pt,shorten <=5pt,thick] (CPoly.south) --  (SS.north) node [midway, right, align=left] (TextNode) {determines \\ essential features};
\end{tikzpicture}
  \caption{Constraint polynomials and spectral structure of the AQRM for $\eta \in \frac12\Z$}
  \label{fig:poly}
\end{figure}
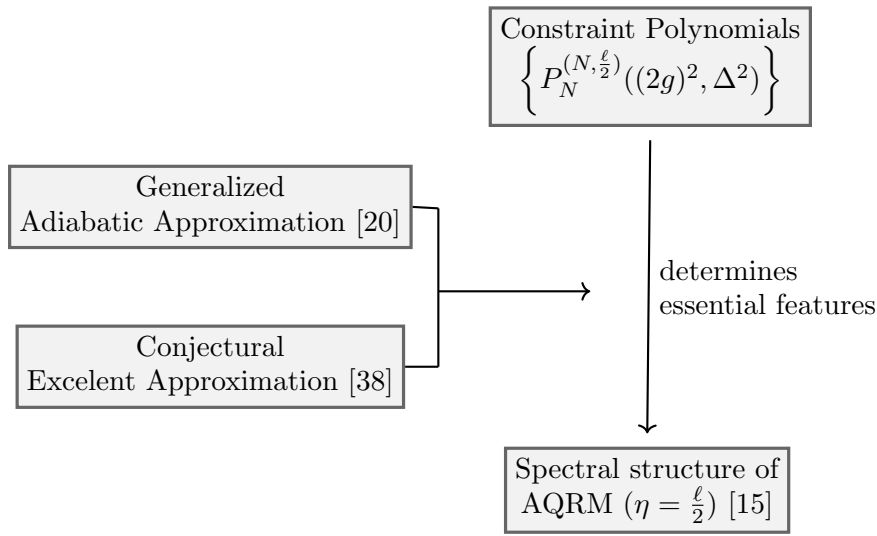

We start by describing the excellent approximation of the lower energy levels of the AQRM in the case $\eta \in \frac12 \Z_{\geq0}$
(see \cite{RBW2021,RBW2022}). As mentioned in Section \ref{sec:preliminaries}, for $\eta \in \frac12\Z_{\ge0}$, the symmetry operator
$J_{2\eta}$ of the AQRM satisfies
\[
  J_{2\eta}^2 = p_{2\eta}(H_{\eta};g,\Delta),
\]
for a polynomial $p_{2\eta}(x;g,\Delta)$. 
It is worth noting that the the equation $y^2=p_\ell(x;g,\Delta)$ defines in general a hyperelliptic curve (an elliptic curve for $2\eta=3,4$). The polynomials $p_\ell(x;g,\Delta)$ are not explicitly known, however it is conjectured
(see \cite{RBW2022}) that
\begin{align*} 
  p_{\ell}(N+\tfrac{\ell}{2} - g^2;g,\Delta) &= A^\ell_N((2g)^2,\Delta^2),
\end{align*}
holds for $\ell,N\geq0$ and $\eta= \frac{\ell}2$, where $A^\ell_N((2g)^2,\Delta^2)$ is a polynomial given by the quotient of the
constraint polynomials for Juddian eigenvalues of the AQRM (see Section \ref{sec:preliminaries}). 
Since it is known that  $A^\ell_N((2g)^2,\Delta^2)$ has a determinant expression \cite{KRW2017}, if the conjecture is assumed to be true,
there is a  determinant expression for $p_{2\eta}(x;g,\Delta)$ for general $x$. This conjecture has been verified numerically for a number
of values of $\ell,N\ge 0$.

Moreover, based on numerical observations it was also conjectured that the curves
\begin{equation}
  \label{eq:polyPandA}
  p_{2\eta}(x;g,\Delta) = 0
\end{equation}
give an excellent approximation of the shape of the first $2\eta$ eigenvalues as shown in Figure \ref{fig:Eigencurves2}. Here, ``excellent approximation'' means, more precisely, that the graphs match well enough to appear to overlap nicely as far as numerical calculations for $2\eta \leq 12$ are concerned. We mention here than in all the observed cases the first $2 \eta$ are non-degenerate.

\begin{figure}[h!]
  \centering
  \includegraphics[height=4.5cm]{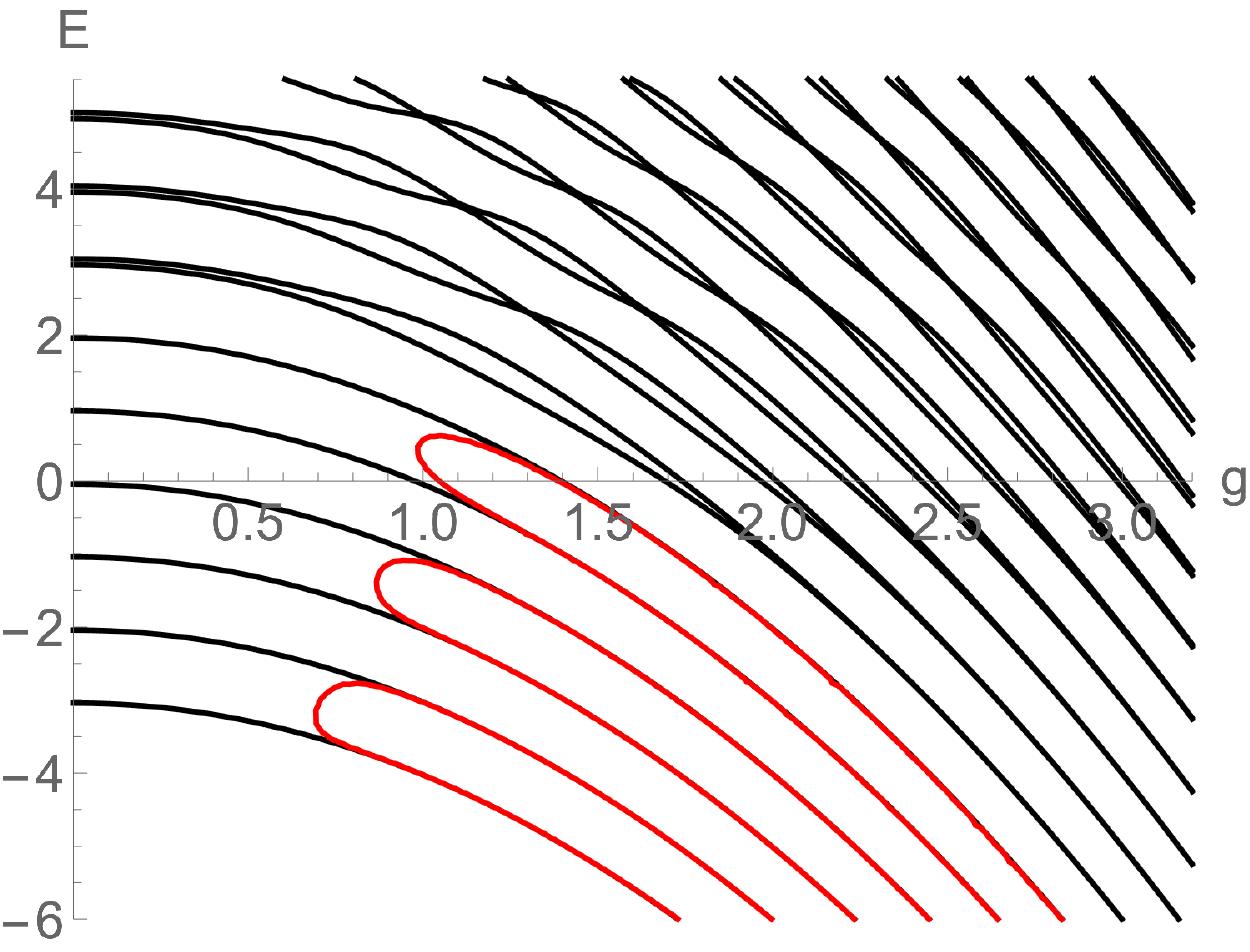}
  \caption{Spectral curves (grey) and curves defined by the determinant expressions \eqref{eq:polyPandA} (red) for $\Delta = \frac12$ and $2\eta=6$.}
  \label{fig:Eigencurves2}
\end{figure}

The nature of the excellent approximation is still mysterious, in particular, it is not yet clear why the spectral curves can be
approximated by the curves \eqref{eq:polyPandA}. Moreover, it is also an open problem to give estimates to the error of the
approximation, that is, to give a qualitative description of the accuracy of the approximation and to investigate whether the energy levels intersect with the standard (baseline) curve $E=N-g^2-\eta$ for some $N \in \Z$ (i.e. whether there are non-Juddian exceptional eigenvalues in the first $2 \eta$ energy levels. We also remark that, in general, it has not been shown that the first $2 \eta$ energy levels for the AQRM with $\eta \in \frac12\Z_{\ge0}$ are multiplicity free.


It is also worth mentioning that the hyperelliptic curves associated to \eqref{eq:hiddenS}, that is, 
\[
  y^2 = p_{2 \eta}(x:g,\Delta),
\]
in addition to the good approximation of eigenvalues, may also be used to introduce a new algebro-geometric
picture of the spectrum. We refer the reader to \cite{RBW2022} for more examples and discussion.

Next, we consider the generalized adiabatic approximation. The adiabatic approximation (AA) is an approximation to the spectral curves of the AQRM using Laguerre polynomials, it is obtained by considered the exact solutions in the extremal case $\Delta \to 0$. The adiabatic approximation is useful for certain parameter regimes of the QRM used in applications, however it is not appropriate for parameter regimes with $g>>1$.

In the paper \cite{LB2021b} the authors present an improved of the AA, called the generalized adiabatic approximation (GAA), obtained by replacing the Laguerre polynomials with the constraint polynomials. 
Concretely, if the energy curve has a crossing, that is, if
\[
  \cp{N,\frac{\ell}{2}}{N}((2g)^2,\Delta^2)=0,
\]
for the corresponding level parameter $N \in \Z$, then the GAA \cite{LB2021b} of the energy curve is
given by
\[
  E^{\ell}_{N,\pm} = N+\frac{\ell}{2} - g^2 \pm  \frac{(-1)^{N+\ell}(2g^2)^\ell \Delta}{2 (N!)^{\frac{3}{2}} \sqrt{(N+\ell)!}}  \exp\left(-2g^2\right) 
  \cp{N,\frac{\ell}{2}}{N}((2g)^2,\Delta^2).
\]
The GAA provides an improved agreement with the exact values than the AA and, in particular, the degenerate points of the GAA
coincide with the exact values of the spectrum (see also the comments in Remark 5.1 in \cite{RBW2022} for a possible refinement). For a wide range of parameter regimes, the general features of the spectrum are captured by the GAA, that is, by the constraint polynomials and the Juddian eigenvalues.

As we have noted, the GAA is only valid for eigenvalue curves with crossings.  Since the polynomial $p_\ell(x; g, \Delta)$ is given by the ratio of two constraint polynomials if the aforementioned conjecture holds, by combining the approximation of the first $\ell$ eigenstates by the curves $p_\ell(x; g, \Delta) = 0$ and the GAA in the $(x, g)$-plane, we see that the constraint polynomials control the essential features of the shape of the spectral curves of AQRM when $\eta \in \frac12 \Z$. Moreover, we also recall that the Juddian solutions, in particular the first Juddian solution for each level, have been shown to be indispensable for the classification of parameter regimes for the QRM in \cite{EVBSS2017}.

\section{Comparison between the $\eta$-NCHO and AQRM}
\label{sec:comparison}

In this paper we have shown that there is a special relation between the two models, the $\eta$-NCHO and the AQRM, via the iso-parallel confluence process. In particular, the relation between the eigenvalues and eigenfunctions
suggests that the $\eta$-NCHO may be considered as a covering (or lifting) of the AQRM. A formal definition of covering model and
further examples is the subject of a forthcoming paper by the authors.
 
In this section we summarize and compare some of the properties of the two models. The main differences are collected in Table \ref{tab:NCHOAQRM}.

\ctable[
caption = {Comparison between the $\eta$-NCHO and AQRM},
label = tab:NCHOAQRM,
pos     = ht,
width   = \hsize,
left
]{l>{\raggedright}X>{\raggedright}X>{\raggedright}X>{\raggedright}X }{
  \tnote[a]{Here, the notion of Heun polynomials does not include usual polynomials.}
  \tnote[b]{Juddian, or quasi-exact, solutions are given by the product of a polynomial and an exponential factor (cf. \eqref{eq:formEigenAQRM})}
  \tnote[c]{Non-Juddian exceptional solutions are eigenfunctions that are not Juddian, that is, they are given by the product of an infinite power series and an exponential factor (cf. \eqref{eq:formEigenAQRM}).}
}{
  \FL
& \multicolumn{2}{c}{$\eta$-NCHO} & \multicolumn{2}{c}{AQRM}
\NN
\cmidrule(r){2-3}\cmidrule(l){4-5}
 & $\eta \notin \frac12 \Z$    & $\eta \in  \frac12 \Z$ & $\eta \notin \frac12 \Z$ & $\eta \in \frac12 \Z$
\NN
\cmidrule(r){2-2}\cmidrule(rl){3-3}\cmidrule(l){4-4} \cmidrule(l){5-5}
Type  & \multicolumn{2}{c}{$\lambda \in \Sigma_0$} & \multicolumn{2}{c}{Exceptional eigenvalues} 
\NN
Form  & \multicolumn{2}{c}{$\lambda = \frac{2 \sqrt{\alpha \beta (\alpha \beta -1)}}{\alpha+\beta} \left(L + \tfrac12 + 2 \eta\right)$}  & \multicolumn{2}{c}{$E = N + \eta -g^2$}
\ML
Multiplicity    & $1$ & $2$ & $1$ & $1$ or $2$
\ML
Eigenfunction & one quasi-exact or one Heun polynomial\tmark[a] & one quasi-exact and one Heun polynomial & one Juddian\tmark[b] or one non-Juddian exceptional\tmark[c] & one non-Juddian exceptional or two Juddian
\ML
Degeneracy &  & same parity & & different parity 
\LL
}

First, it is important to emphasize again that the $\eta$-shift term of the $\eta$-NCHO does not break the
$\Z_2$-symmetry of the NCHO. In fact, the parity operator $ \bm{I}_2 \mathcal{P}$ commutes with $Q^{(\eta)}$ for any $\eta$. This is in contrast with the AQRM, where the $\Z_2$-symmetry operator $J_0 =  \mathcal{P} \sigma_z$ of the QRM does not commute with the Hamiltonian of the AQRM except for QRM case (i.e. $\eta=0$).

For the AQRM, the symmetry breaking effect of the bias term $\eta \sigma_x$ makes the spectrum multiplicity free
for general $\eta$, except for the case $\eta \in \frac12 \Z$ where a hidden symmetry operator $J_{2\eta}$ is known to
exist. In the $\eta \in \frac12 \Z$ case, the exceptional eigenvalues of the AQRM, that is, eigenvalues of the form
\[
  E = N + \eta -g^2
\]
for $N \geq 0$ are shown to be either multiplicity free or degenerate with multiplicity $2$ The degenerate case always occurs between two Juddian solutions and, conversely, any degeneration of the spectrum of the AQRM is of this type. In the case where the exceptional eigenvalue is not degenerate, it always consist of a non-Juddian exceptional solution. For the proof of these results and an extended discussion of the properties of the spectrum of the AQRM, we refer the reader to \cite{KRW2017}. Therefore, the Juddian and non-Juddian exceptional solutions of the AQRM may be considered the equivalent of the quasi-exact eigenfunctions for eigenvalues in $\Sigma_0$ and Heun polynomial holomorphic (i.e. not usual polynomials) eigenfunctions for eigenvalues in $\Sigma_0 \cap \Sigma_\infty$, respectively. 

In the case of the $\eta$-NCHO, the eigenvalues $\lambda \in \Sigma_0$ are of the form
\[
  \lambda = \frac{2 \sqrt{\alpha \beta (\alpha \beta -1)}}{\alpha+\beta} \left(L + \tfrac12 + 2 \eta\right),
\]
for $L \geq 0$. In the general case, the eigenvalues of the $\eta$-NCHO of this form have multiplicity $1$ and the corresponding eigenfunction is either quasi-exact solution or a Heun polynomial. However, when $\eta \in \frac12 \Z$ the eigenvalues $\lambda \in \Sigma_0$ are degenerate with multiplicity $2$ and the eigenfunctions consist of a quasi-exact solution (i.e., polynomial solution) and a Heun polynomial. In this way, despite the differences with respect to their corresponding symmetry operators, the situation with respect to degeneracy for both the $\eta$-NCHO and the AQRM is similar. We also note that the $\eta$-NCHO may have a different type of spectral degeneracy (of multiplicity $2$) for $\lambda \in \Sigma_\infty^+\cap\Sigma_\infty^-$ that does not have an analog on the spectrum of the AQRM (see the discussion at the end of Section~\ref{sec:multiplicity}).

It is important to notice that for the case $\eta \in \frac12 \Z$, under the iso-parallel confluence process, the eigenvalues of the $\eta$-NCHO that descend into a degenerate Juddian eigenvalue have the same parity but correspond to different eigenvalue problems of $Q^{(\eta)}$ and $\bK Q^{(\eta)} \bK$. Further,  by Proposition \ref{prop:confl}, the confluent ODE picture given by the iso-parallel confluence process for the Heun picture of $Q^{(\eta)}$ (resp. $\bK Q^{(\eta)} \bK$)  gives the development of the AQRM eigenvalue problem at the singularity $-g$ (resp. $g$), that is, the confluent Heun ODE \eqref{eq:H1eps} (resp. \eqref{eq:H2eps}). In addition, the eigenfunctions that remain under the confluence process are those that correspond to quasi-exact solutions in the $\eta$-NCHO and the Heun polynomials eigenfunctions vanish. In this sense, the hidden symmetry of the AQRM may be originated and explained by the facts that ${\rm \Spec}(Q^{(\eta)}) = {\rm \Spec}( \bK Q^{(\eta)} \bK)$ and that $Q^{(\eta)}$ is parity invariant.

\appendix

\section{Heun and confluent Heun differential equations}
\label{sec:HeunODE}

Let us recall the basic notions and notations for Heun and confluent Heun differential equations. In general, we follow standard
notations (see e.g. \cite{R1995,SL2000}).

The general confluent Heun differential equation (ODE) is of the form
\begin{equation}
  \label{eq:HeunODE}
  \frac{d^2 y}{d z^2} + \left( \frac{\gamma}{z} + \frac{\delta}{z-1} + \frac{\epsilon}{z-a} \right) \frac{d y}{d z} + \frac{\alpha \beta z - q}{z(z-1)(z-a)} y = 0,
\end{equation}
for $a \neq 0,1$ and parameters $\alpha, \beta,\gamma,\delta,\epsilon,q \in \C$ satisfying
\[
  \gamma + \delta + \epsilon = \alpha + \beta + 1.
\]
The differential equation \eqref{eq:HeunODE} is of Fuchsian type with regular singularities at $z=0,1,a,\infty$ with exponents
\[
  \{0, 1- \gamma, \}, \qquad \{0,1-\delta\}, \qquad \{0,1-\epsilon\}, \qquad \{\alpha, \beta\},
\]
respectively. The information of the singularities is summarized in the Riemann scheme given by
\[
  \begin{pmatrix}
    0 & 1 & t & \infty &; z \\
    0 & 0 & 0 & \alpha&; q \\
    1-\gamma & 1-\delta & 1-\epsilon & \beta &
  \end{pmatrix}.
\]

The solutions of the Heun ODE are given by the Heun functions, defined in regions containing two singularities and some care must
be taken to ensure single valuedness. In this paper however we only consider the solutions given by Heun polynomials
\[
  Hp(z) = z^{\sigma_1} (z-1)^\sigma_2 (z-a)^{\sigma_3} p(z),
\]
for some polynomial $p$ in $z$. In general we will consider only the solutions of type I, where $\sigma_i=0$ for $i=1,2,3$, where
a necessary condition for the existence of these polynomials is that $\alpha = -N$ for $N \in \Z_{\geq 0}$. An additional consistency condition must be added to ensure the existence of solutions.

The general confluent Heun ODE with regular singularities at $z=0,1$ is given by
\begin{equation}
  \label{eq:CHeunODE}
  \frac{d^2 y}{d z^2} + \left(4 p + \frac{\gamma}{z} + \frac{\delta}{z-1}\right) \frac{d y}{d z} + \frac{ 4 p \alpha z - \sigma}{z(z-1)} y = 0,
\end{equation}
with parameters $\gamma,\delta,p,\alpha \in \C$ and accessory parameter $\sigma \in \C$.

The generalized Riemann scheme of the confluent Heun differential equation is given as follows
\[
  \begin{pmatrix}
    1 & 1 & 2 & \\
    0 & 1 & \infty &; \phantom{-}z\\
    1-\gamma& 1-\delta & \gamma+\delta-\alpha &; -\sigma \\
      &     & 0 & \\
      &     & 4p  &
  \end{pmatrix}.
\]

\section{Limit curves determined by the constraint polynomials of $\eta$-NCHO}
\label{sec:curvesNCHOAQRM}

In this appendix we describe the relation between the curves determined by the constraint polynomials of $\eta$-NCHO and AQRM under the process of iso-parallel confluence. In Theorem \ref{thm:constConflucence2} we showed that via the iso-parallel confluence, the functions  $q^{(L,\eta,a)}_i(\alpha,\beta)$ (defined in \eqref{eq:normCP}) associated to the $\eta$-NCHO correspond to the polynomials $\cp{N,\eta}{i}((2g)^2,\Delta^2)$ of certain AQRM.
In this section, we illustrate this result using the plane curves as $\alpha \beta \to \infty$.

In order to compare the two curves between the two models it is necessary to get an expression for the functions $q^{(L,\eta,a)}_i(\alpha,\beta)$
that contains the system parameters of the AQRM. We start from \eqref{eq:constPolysNCHO} and apply the change of variable $a\to a+p$
and $L \to L+p$ and replace the asymptotic relation
\[
  \left|\frac{\alpha-\beta}{\alpha+\beta}\right| = k(\alpha \beta)^{-1} + o\left((\alpha \beta\right)^{-2})
\]
of the iso-parallel confluence process (see Definition \ref{dfn:ipc}) to obtain
\begin{align}
  \label{eq:polyqAs}
  \tilde{q}^{(L,\eta,a)}_i(\alpha\beta,r,k)  &=  \tfrac14 \tilde{d}_{L-2i}^{(L)} \tilde{q}_{i-1}^{(L,\eta,a)}(\alpha\beta,r,k) \nonumber \\
  &\, - \frac12\left(\frac{L-2i+a}{2\alpha \beta} + r + o(1)\right)(L-2i-a+3)i(i-1) \tilde{q}_{i-2}^{(L,\eta,a)}(\alpha\beta ,r,k), 
\end{align}
with
\begin{align*}
  \tilde{d}_{L-2i}^{(L)}  &= \frac{\alpha\beta-1}{\alpha\beta} ( k + o(1))^2 \left(\frac{2 L + 1 + 4 \eta}{2 \alpha \beta} + r + o(1)\right)^2  \\ & \qquad - 4 i \left\{(i+ 2\eta-r)  -  \frac{2L-2 i + 1 + 4\eta}{2\alpha \beta} + o(1)\right\},
\end{align*}
where the functions of order $o(1)$ are not necessarily the same. Note that $\tilde{q}^{(L,\eta,a)}_i(\alpha\beta,r,k)$ depends only on $\alpha \beta$ and not on $\alpha/\beta$ (up to possible $o(1)$ terms).
For simplicity, in this section we take the $o(1)$ terms to be
equal to the constant function $0$ and furthermore we take $a= -L+1$ (in this case $L=N$ in Theorem \ref{thm:constConflucence2}).

After the change of variable given in Theorem \ref{thm:constConflucence2}, we obtain the expression
\begin{align}
  \label{eq:constTrans}
    \tilde{q}^{(L,\eta,a)}_i(\alpha\beta,g,\Delta)=0
\end{align}
describing a surface in $(\alpha \beta,g,\Delta)$--space. For simplicity of visualization, we consider projections in the $(g,\Delta)$ plane for different
values of $y=\alpha\beta$. As we can see in Figures \ref{fig:transcurve1} and \ref{fig:transcurve2}, as $\alpha \beta \to \infty$ the curves approach the
curve given by the constraint relation for the AQRM shown in Figure \ref{fig:transcurve3}.

\begin{figure}[h]
  \centering
  \subfloat[$\alpha\beta=4$]{
    \includegraphics[height=4cm]{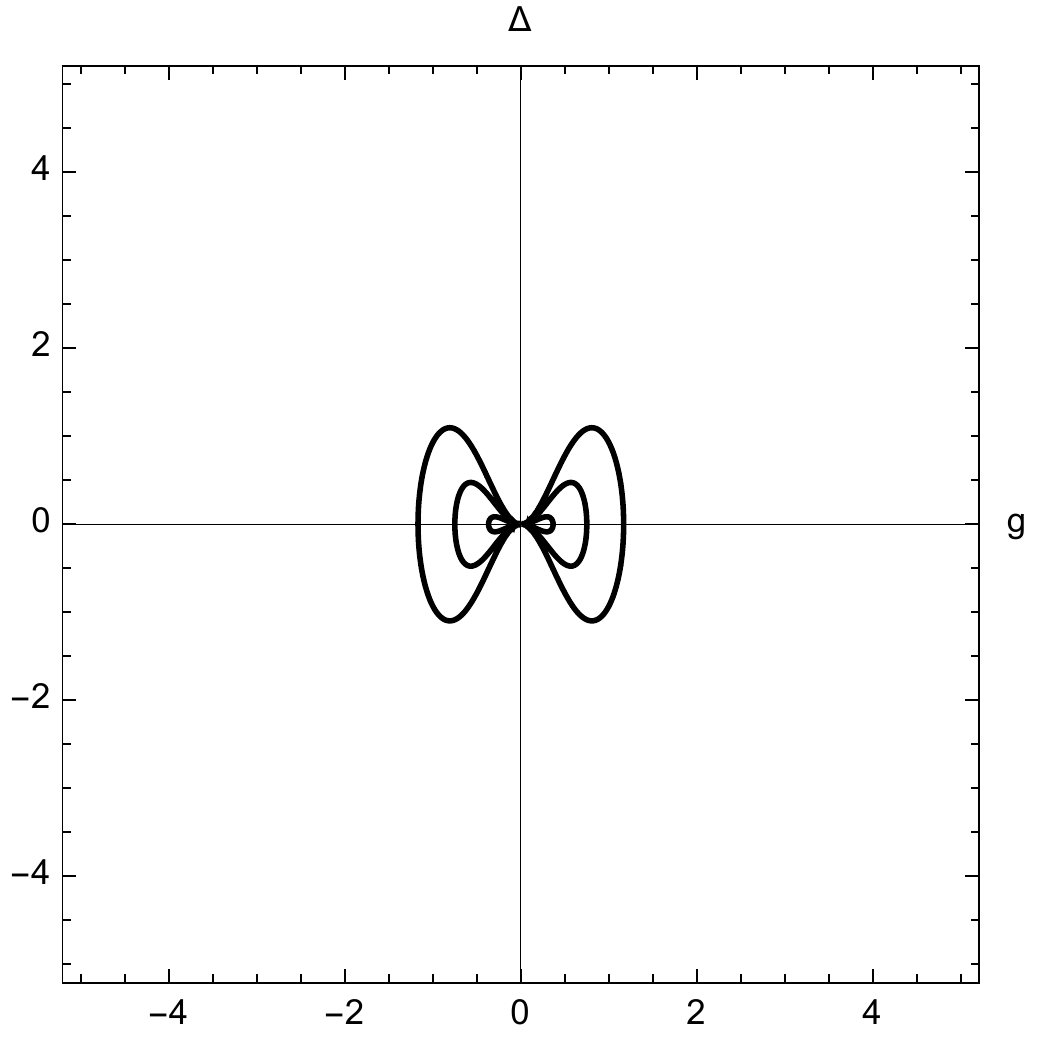}}
  ~ \qquad \qquad
  \subfloat[$\alpha\beta=2^4$]{
    \includegraphics[height=4cm]{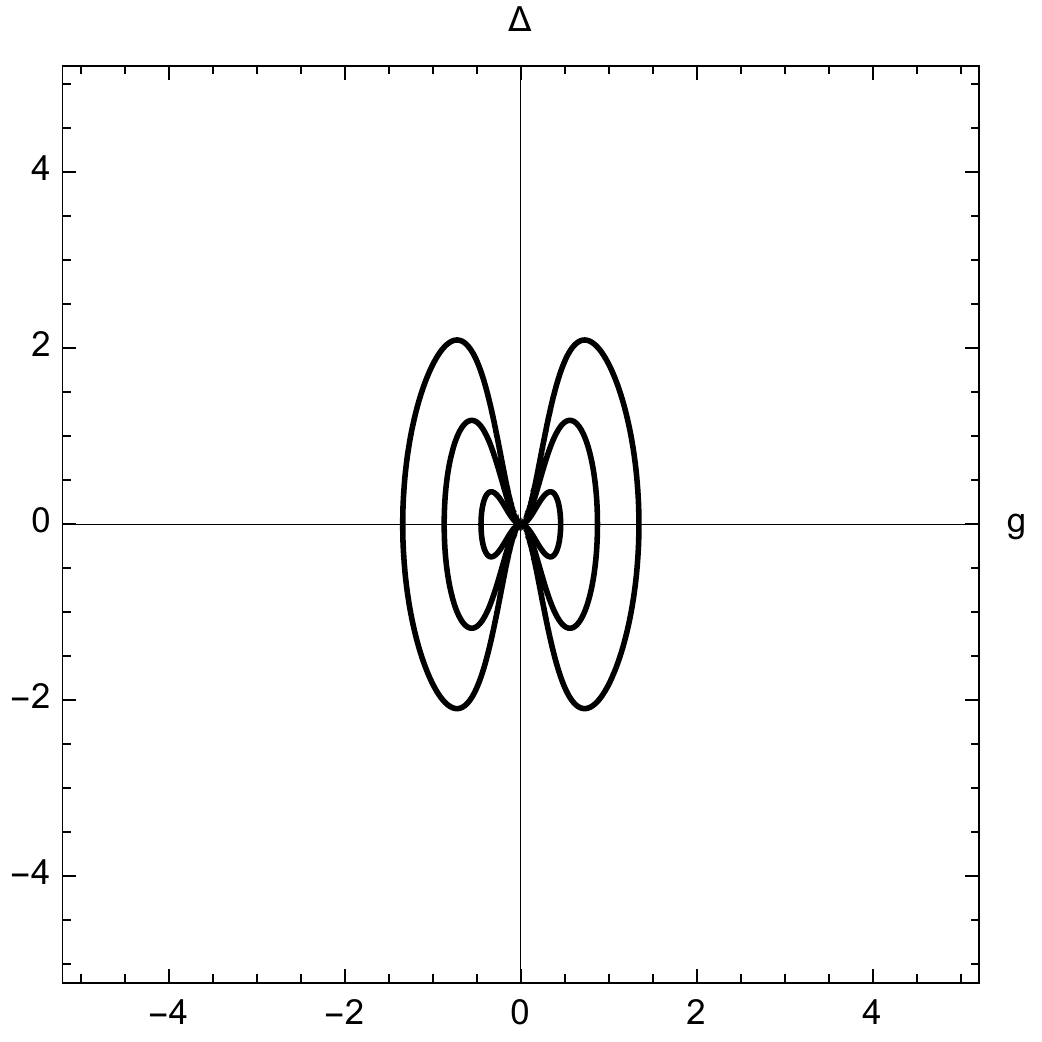}}
  \caption{Curves determined by \eqref{eq:constTrans} for $\eta=\frac12$}
  \label{fig:transcurve1}
\end{figure}

\begin{figure}[h]
  \centering
  \subfloat[$\alpha \beta=2^8$]{
    \includegraphics[height=4cm]{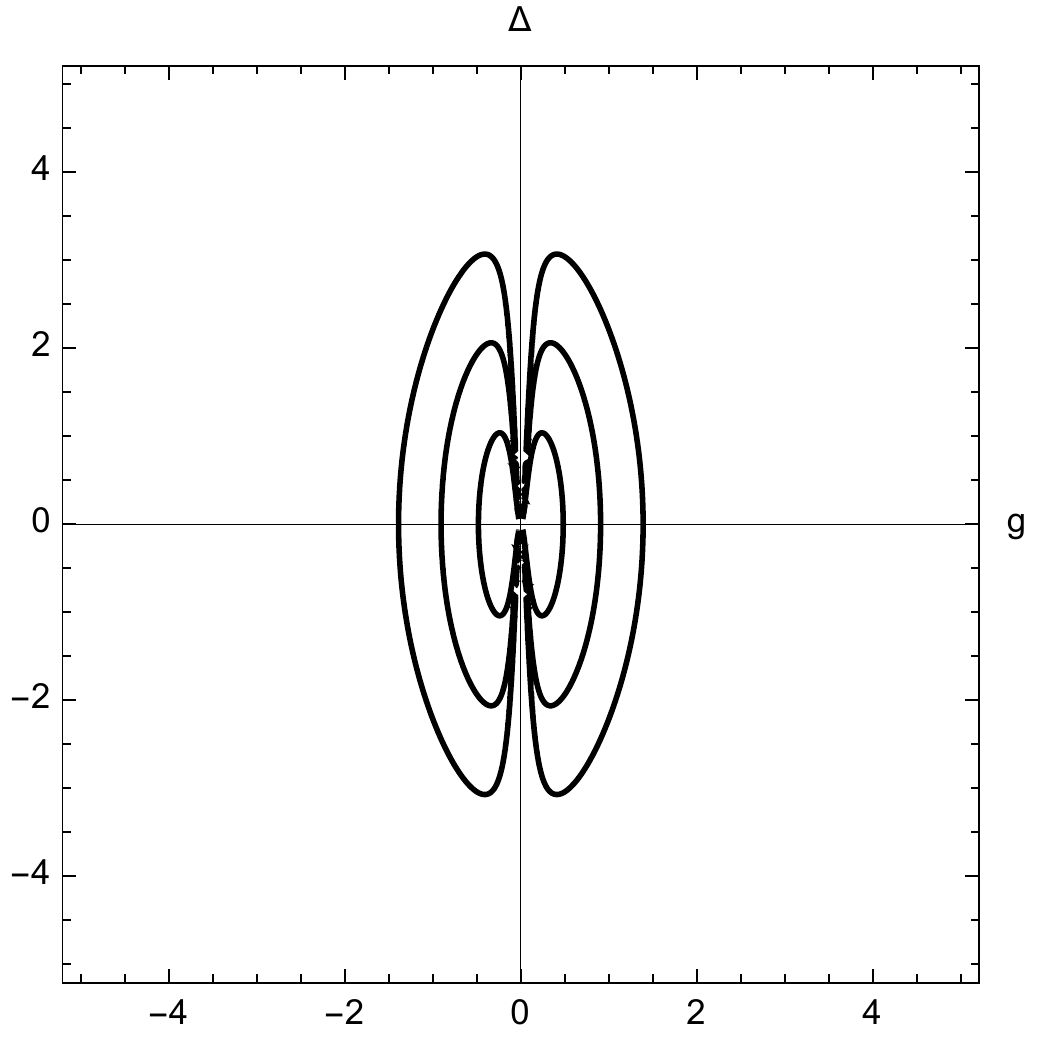}}
  ~ \qquad \qquad
  \subfloat[$\alpha\beta=2^{10}$ ]{
    \includegraphics[height=4cm]{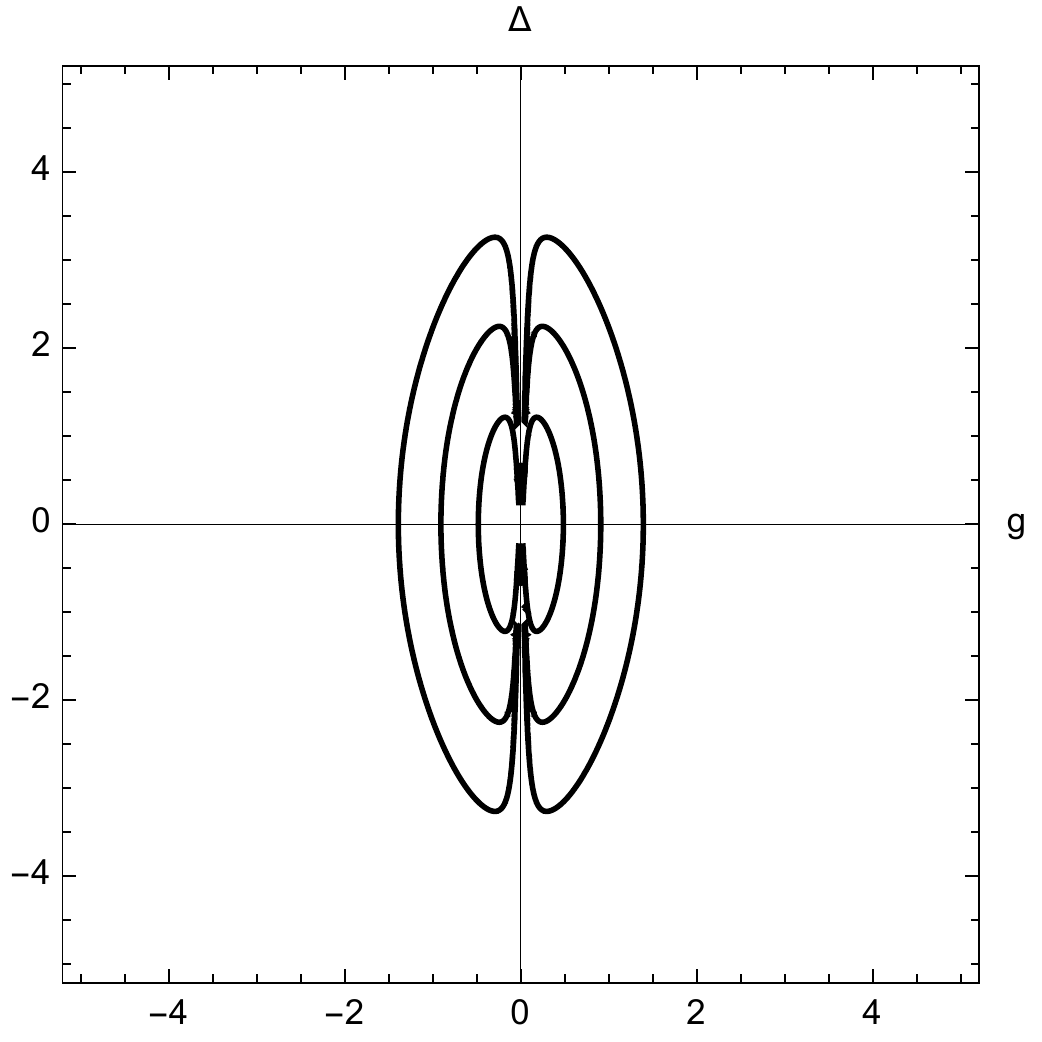}}
  \caption{Curves determined by \eqref{eq:constTrans} for $\eta=\frac12$}
  \label{fig:transcurve2}
\end{figure}

\begin{figure}[h]
  \centering{
    \includegraphics[height=4cm]{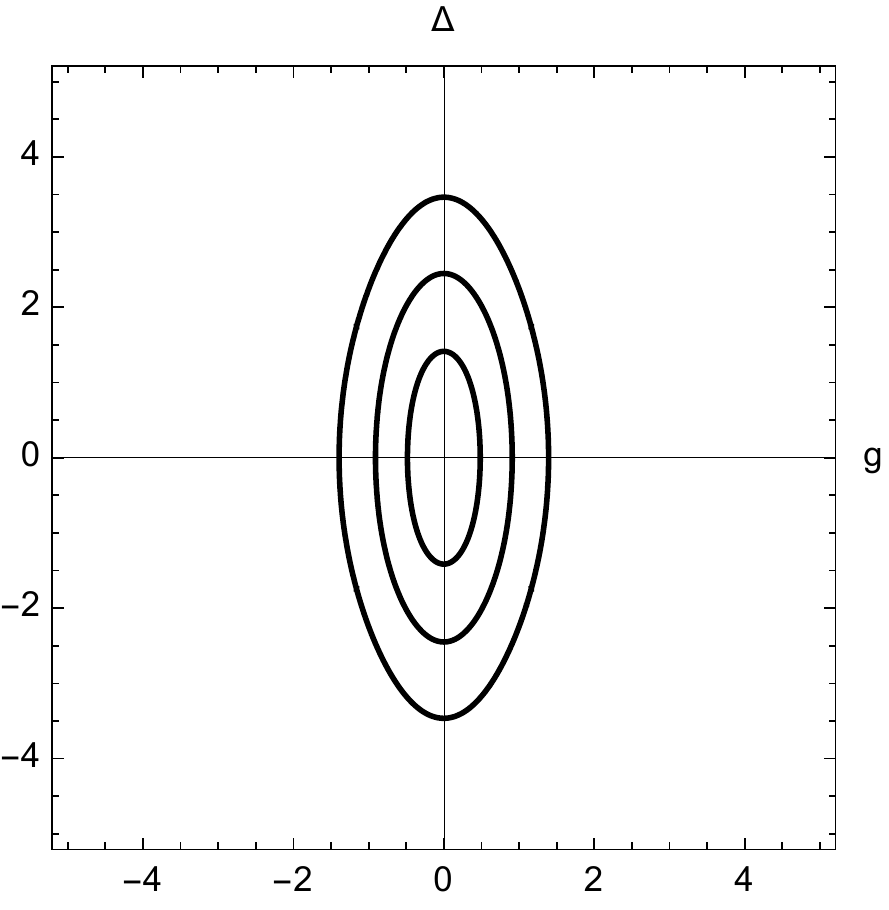}}
  \caption{Curves determined by the constraint condition for the AQRM with $\eta=\frac12$ and $N=3$}
  \label{fig:transcurve3}
\end{figure}

\section*{Acknowledgements}

This work was partially supported by JSPS Grant-in-Aid for Scientific Research (C) No.20K03560,  JST CREST JPMJCR14D6 and CREST JPMJCR2113, Japan.

\begin{flushleft}
  
\bigskip

  Cid Reyes-Bustos \\
  NTT Institute for Fundamental Mathematics,\\
  NTT Communication Science Laboratories, NTT Corporation \\
3-9-11, Midori-cho Musashino-shi, Tokyo, 180-8585, Japan \\
  email: \texttt{reyes.cid.zr@hco.ntt.co.jp, math@cidrb.me}
  \phantom{a}\\\phantom{a}\\

  Masato Wakayama \\
  NTT Institute for Fundamental Mathematics,\\
  NTT Communication Science Laboratories, NTT Corporation \\
  3-9-11, Midori-cho Musashino-shi, Tokyo, 180-8585, Japan \\
  and\\
  Institute for Mathematics for Industry,  Kyushu University, \\
  744 Motooka, Nishi-ku, Fukuoka 819-0395\, Japan\\
  email: \texttt{wakayama@imi.kyushu-u.ac.jp, masato.wakayama.hp@hco.ntt.co.jp}
  
\end{flushleft}

\end{document}